\documentclass[journal, compsoc, final, twoside,10pt]{IEEEtran}

\usepackage{units, graphbox, multicol, amsmath,
cellspace, multirow, textcomp, wrapfig, verbatim,url,subfig,booktabs,dblfloatfix,ntheorem,mathtools,amssymb,tabulary,cite,algorithm,algorithmic,RobStd}
\usepackage{tikz}
\usetikzlibrary{patterns,calc}
\usepackage{pgfplots}
\pgfplotsset{
    table/search path={Simulations},
    compat=1.14
}
\usetikzlibrary{automata, positioning}
\usepackage{flushend}
\usepackage{prettyref}
\newcommand{\pref}{\prettyref}
\usepackage{xspace}

\newcommand{\hu}{}
\newcommand{\sh}{}

\makeatletter

\newrefformat{eq}{\textup{(\ref{#1})}}
\newrefformat{lem}{Lemma~\ref{#1}}
\newrefformat{thm}{Theorem~\ref{#1}}
\newrefformat{sec}{Section~\ref{#1}}
\newrefformat{tab}{Table~\ref{#1}}
\newrefformat{fig}{Fig.~\ref{#1}}
\newrefformat{alg}{Algorithm~\ref{#1}}

\graphicspath{{images/}{Figures/}}
\renewcommand{\arraystretch}{1.2}

\floatstyle{ruled}
\newfloat{algorithm}{tbp}{loa}
\floatname{algorithm}{Algorithm}

\renewcommand{\IEEEQED}{\IEEEQEDopen}

\newcommand{\fpas}{{FD-PaS}\xspace}
\newcommand{\dpas}{{D$^2$-PaS}\xspace}
\newcommand{\prob}{{\bf Problem 1}}

\newcommand{\tset}{\mbox{$\mathcal{T}$}}
\newcommand{\dropset}{\rho}

\newcommand{\rhymode}{rhythmic mode}
\newcommand{\normode}{nominal mode}
\newcommand{\rhystate}{rhythmic state}
\newcommand{\norstate}{nominal state}

\newcommand{\greedy}{greedy heuristic}

\newcommand{\tep}{\ensuremath{t_{ep}}}

\newcommand{\tnr}{\ensuremath{t_{n\rightarrow r}}}
\newcommand{\trn}{\ensuremath{t_{r\rightarrow n}}}
\newcommand{\dynamic}{\ensuremath{\tilde{S}}}
\newcommand{\static}{\ensuremath{S}}

\newcommand{\taui}{$\tau_i \; (1 \leq i \leq n)$}
\newcommand{\pkt}{\chi}
\newcommand{\pend}{{\bf Problem 1.1}}
\newcommand{\pdrop}{{\bf Problem 1.2}}
\newcommand{\dnodes}{$\mathbf{V}_{rhy}$}
\newcommand{\slotset}{E}
\newcommand{\slotvalue}{\epsilon}

\newcommand{\xset}{\ensuremath{\mathcal{S}}}

\newcommand{\retryfunc}{$\overrightarrow{R}^*_i(\cdot)$}
\newcommand{\pdrfunc}{$\lambda^*_i(\cdot)$}

\newtheorem{lemma}{Lemma}
\newtheorem{theorem}{Theorem}
\newtheorem{constraint}{Constraint}

\newcommand{\eat}[1]{}

\newcommand{\han}{}
\newcommand{\tz}{}
\newcommand{\extend}{}
\newcommand{\cam}{}

\newcommand{\pmac}{MP-MAC}

\makeatother
\IEEEoverridecommandlockouts
\begin{document}

\title{\LARGE Fully Distributed Packet Scheduling Framework for Handling Disturbances in Lossy Real-Time Wireless Networks}

\author{Tianyu~Zhang, Tao~Gong, Song~Han,~\IEEEmembership{Member,~IEEE,} Qingxu~Deng,~\IEEEmembership{Member,~IEEE,} Xiaobo~Sharon~Hu,~\IEEEmembership{Fellow,~IEEE} %
\IEEEcompsocitemizethanks{\IEEEcompsocthanksitem T. Zhang is with the Qingdao University, China. Email: tzhang4@nd.edu.
\IEEEcompsocthanksitem T. Gong and S. Han are with the University of Connecticut, Storrs, CT, 06269. Email: \{tao.gong,song.han\}@uconn.edu.
\IEEEcompsocthanksitem X. S. Hu is with the University of Notre Dame, Notre Dame, IN, 46556. Email: shu@nd.edu.
\IEEEcompsocthanksitem Q. Deng is with the Northeastern University, Shenyang 110819, China. Email: dengqx@mail.neu.edu.cn.}
\thanks{The first two authors have equal contribution to this work.}
}



\IEEEtitleabstractindextext{
\begin{abstract}

Along with the rapid growth of Industrial Internet-of-Things (IIoT) applications and their penetration into many industry sectors, real-time wireless networks (RTWNs) have been playing a more critical role in providing real-time, reliable and secure communication services for such applications. A key challenge in RTWN management is how to ensure real-time Quality of Services (QoS) especially in the presence of {\han unexpected disturbances} and lossy wireless links. Most prior work takes centralized approaches for handling disturbances, which are slow and subject to single-point failure, and do not scale. To overcome these drawbacks, this paper presents a fully distributed packet scheduling framework called \fpas{}. \fpas{} aims to provide guaranteed fast response to unexpected disturbances while achieving {\han minimum performance} degradation for meeting the timing and reliability requirements of all critical tasks. To combat the scalability challenge, \fpas{} incorporates several key advances in both algorithm design and data link layer protocol design to enable individual nodes to make on-line decisions locally without any centralized control. Our extensive simulation and testbed results have validated the correctness of the \fpas{} design and demonstrated its effectiveness in providing fast response for handling disturbances {\han while ensuring the designated QoS requirements}. 

\end{abstract}
}

\maketitle

\IEEEpeerreviewmaketitle
\setlength{\textfloatsep}{0.5\baselineskip}
\setlength{\intextsep}{0.5\baselineskip}
\section{Introduction}\label{sec:intro}
\eat{
1. Why we study RTWN: discuss their applications, etc.

2. The challenges of studying RTWN including (i) the explosive growth of IoT applications especially in terms of their scale and complexity, and (ii) the fact that almost all RTWNs must deal with unexpected disturbances.

3. What have been done on tackling these challenges and what are their downsides.
\begin{itemize}
    \item Static approach cannot handle unexpected disturbances or must make rather pessimistic assumptions. 
    \item Centralized approach (the state-of-the-art OLS) has very high computation overhead and relatively low performance.
    \item Distributed approach (\dpas{}) relies on at least one single point in the network to make on-line decisions and therefore has limited scalability.
\end{itemize}

4. A big picture of what problem we studied in this work.

5. The contributions and organization of this paper.
}

Real-time wireless networks (RTWNs) are fundamental to many Industrial Internet-of-Things (IIoT) applications in a broad range of fields such as military, civil infrastructure and industrial automation~\cite{Hei_INFOCOM13, Gatsis_ACC13, Karbhari_Elsevier09}. These applications have stringent {\extend timing and reliability requirements} to ensure timely collection of environmental data and {\han reliable} delivery of control decisions. The Quality of Service (QoS) offered by a RTWN is thus often measured by how well it satisfies the end-to-end (from sensors via controllers to actuators) deadlines of the real-time tasks executed in the RTWN. Packet scheduling in RTWNs plays a critical role in achieving the desired QoS. Though packet scheduling in RTWNs has been studied for a long time, the explosive growth of IIoT applications especially in terms of their scale and complexity has dramatically increased the level of difficulty in tackling this inherently challenging undertaking. The fact that most RTWNs must deal with unexpected disturbances {\extend and the {\han lossy nature of} wireless links in industrial environments further aggravates the problem.}

{\extend
Unexpected disturbances {\han in RTWNs in general} can be classified into {\em internal} disturbances within the network infrastructure ({\em e.g.}, link failure due to multi-user interference or weather related changes in channel signal to noise ratio (SNR)) and {\em external} disturbances from the environment being monitored and controlled ({\em e.g.}, detection of an emergency, sudden pressure or temperature changes). 
When an external disturbance is detected by a certain sensor node, the workload associated to this sensor node needs to be changed for a certain time duration to more frequently monitor the environment.
Many centralized dynamic scheduling approaches have been proposed in the literature, but most of them are designed for handling changes in network resource supply ({\em e.g.}, \cite{Crenshaw_TECS07, Shen_WN13, ferrari2012low}). Studies on addressing external disturbances in RTWNs, the focus of this paper, are relatively few. {Most of those work rely on centralized decision making and assume reliable network environments. This motivates us to explore a fully distributed framework for handling external disturbances in lossy RTWNs.}
In the rest of the paper, we simply refer to external disturbance as disturbance.

The challenge of handling disturbances in RTWNs comes from the unpredictability of disturbance occurrence at run time. Specifically, it is generally unknown when/which disturbance {\han will occur} and what is the network status at that point ({\em e.g.}, how many packets have been delivered to their destinations). Since it is computationally infeasible to enumerate all possibilities before the network starts, on-line dynamic scheduling approaches is required to react {\han fast} to unexpected workload changes incurred by disturbances.
}

{\extend The existence of lossy wireless links in the industrial environments raises another challenge in handling disturbances in RTWNs. Specifically, the uncertainty of lossy links in the network introduces packet losses with a certain non-zero possibility. Packet loss in a sensing process can significantly degrade the data freshness, and packet loss in a feedback control may lead to system instability and cause safety concerns. Further, if a packet that delivers disturbance-related information is lost, it may cause catastrophe to the system. Thus, most industrial RTWNs require a desired end-to-end Packet Delivery Ratio (PDR), {\em e.g.} {\han $99\%$}, for all packets running in the system.}

In this work, we introduce a fully distributed packet scheduling framework, referred to as \fpas{}, to handle disturbances in lossy RTWNs.\footnote{An earlier version of the paper appeared in \cite{fdpas}.} 
\fpas{} makes on-line decisions locally without {\em any} centralized control point 
when disturbances occur. This is achieved by sending the disturbance information only to a subset of all nodes via the routing paths of the tasks running in the network. 
In such a manner, a broadcast task is no longer needed in \fpas{} for notifying all nodes about the disturbance information, which significantly reduces the response time to handle the disturbance. To ensure this partial disturbance propagation scheme works properly, we need to overcome several challenges. For example, to avoid transmission collision among different nodes with inconsistent schedules, we propose a multi-priority wireless packet preemption mechanism called {\pmac{}} in the data link layer to ensure that high-priority packets can always be delivered by preempting the transmissions of low-priority packets.
{\extend Further, to minimize the timing and reliability degradation, we formulate a transmission dropping problem to determine a temporary dynamic schedule for individual nodes to handle the disturbance.} We prove that the transmission dropping problem is NP-hard, and introduce an efficient heuristic to be executed by individual nodes locally. Both the \pmac{} design and the dynamic schedule construction method (they jointly comprise the FD-PaS framework) are implemented on our RTWN testbed. Our extensive performance evaluation validates the correctness of the FD-PaS design and demonstrates its effectiveness in providing fast response for handling disturbances. 

{\extend 
}

The remainder of this paper is organized as follows. The related work is dicussed in Section~\ref{sec:related} and Section~\ref{sec:model} describes the system model. Section~\ref{sec:framework} gives an overview of the \fpas{} framework. We discuss how to propagate disturbances and avoid transmission collisions in Section~\ref{sec:propagating} and~\ref{sec:collisions}, respectively. 
{\extend Section~\ref{sec:dropping} formulates the dynamic transmission dropping problem and presents the method to determine the time duration for handling disturbance. Section~\ref{sec:dynamic} discusses the dynamic schedule generation in both reliable and lossy RTWNs.} Performance evaluation are summarized in Section~\ref{sec:simulation}. We conclude the paper and discuss future work in  Section~\ref{sec:conclusions}.

\section{Related Work}\label{sec:related}

Network resource management in RTWNs in the presence of unexpected disturbances has drawn a lot of attention in recent years. Traditional static packet scheduling approaches ({\em e.g.}, \cite{Han_RTAS11, Leng_RTSS14, Saifullah_RTSS10}), where decisions are made offline or only get updated infrequently can support deterministic real-time communication, but either cannot properly handle unexpected disturbances or must make rather pessimistic assumptions. Many centralized dynamic scheduling approaches for handling internal disturbances have been proposed ({\em e.g.}, \cite{Crenshaw_TECS07, Shen_WN13, ferrari2012low}). Studies on addressing external disturbances are relatively few and mostly rely on centralized decision making. 
The approach in~\cite{Sha_RTSS13} stores a predetermined number of link layer schedules in the system and chooses the appropriate one when disturbances are detected. However, this approach is either incapable of handling arbitrary disturbances or needs to make some approximation. 
Both \cite{Chipara_ECRTS11} and \cite{zimmerling2017adaptive} support admission control in response to adding/removing tasks for handling disturbances in the network. They however do not consider scenarios when not all tasks can meet their deadlines.
The protocol in~\cite{Li2015Incorporating} proposes to allocate reserved slots for occasionally occurring emergencies (i.e., disturbances), and allow regular tasks to steal slots from the emergency schedule when no emergency exists. However, how to satisfy the deadlines of regular tasks in the presence of emergencies is not considered. 

In recent years, a number of algorithms have been designed for packet scheduling in Time Slotted Channel Hopping (TSCH) networks, in both centralized ({\em e.g.} \cite{palattella2013optimal,soua2012modesa,soua2013musika}) and distributed manner ({\em e.g.} \cite{tinka2010decentralized,morell2013label,soua2016wave}). Most of those approaches, however, assume static network topologies and fixed network traffic which limit their applications in dynamic networks. To overcome this drawback, \cite{duquennoy2015orchestra} proposes Orchestra, a distributed scheduling solution that schedules packet transmissions in TSCH networks to support real-time applications. 
However, Orchestra does not consider real-time constraint, {\em i.e.}, ignores the hard deadlines associated with tasks running in the network. It only provides best effort but no guarantee on the end-to-end latency of each task.

In \cite{hong2015online}, a centralized dynamic approach, named OLS, to handle disturbances in RTWNs is proposed. 
OLS is built on a dynamic programming based approach which can be rather time consuming even for relatively small RTWNs. Moreover, OLS may drop more periodic packets than necessary due to the limited payload size of the packet in RTWNs and thus further degrade the system performance. 
To overcome the drawbacks of OLS, \dpas{} in \cite{zhang2017distributed,zhang2018distributed} proposes to offload the computation of the dynamic schedules to individual nodes locally 
by leveraging their local computing capabilities, that is, letting each node construct its own schedule so as to achieve better performance than OLS in terms of fewer dropped packets and lower time overhead. However, as observed from the motivating example presented in \cite{fdpas}, centralized approaches, including {\dpas{}}, suffer from long disturbance response time especially in large RTWNs.

Most MAC layer designs for supporting packet prioritization are based on star topology. For example, the wireless arbitration (WirArb) method~\cite{zheng2016wirarb} is designed to use different frequencies to indicate different priorities. It only supports star topology where the gateway keeps sensing the arbitration signals and determines which user has a higher priority to access the channel. \cite{shao2014multi} studies a similar problem in the context of vehicular Ad Hoc networks. The proposed multi-priority MAC protocol has seven channels, among which one is the public control channel (CCH) for safety action messages and the others are service channels for non-safety applications. The protocol transmits packets of different priorities with optimal transmission probabilities in a dynamic manner. The PriorityMAC~\cite{prioritymac} proposes to add two very short sub-slots before each time slot to indicate the priority. Four priority levels are defined but only three levels of over-the-air preemption can be achieved. The last priority level is only used for buffer reordering. In PriorityMAC, a higher priority packet indicates the priority in the sub-slots to deter the transmissions of lower priority packets. PriorityMAC is also based on star topology so each device must be directly connected to the coordinator.

{\extend 
A rich set of methods have been designed for RTWNs to improve the reliability of wireless packet transmission over lossy links in most RTWN solutions ({\em e.g.}, WirelessHART~\cite{song2008wirelesshart}, ISA 100.11a~\cite{isa10011a}, and 6TiSCH~\cite{dujovne20146tisch}). \cite{Han_RTAS11} proposed a set of reliable graph routing algorithms in WirelessHART networks to explore path diversity to improve reliability. 
\cite{probabilistic, flexible} proposed algorithms to allocate a necessary number of retransmision time slots to guarantee a desired success ratio of packet delivery. 
However, all aforementioned studies focus on packet scheduling in static RTWN settings over lossy links, and cannot be easily extended to handle abruptly increased network traffic caused by unexpected disturbances.
}
\section{System Model}
\label{sec:model}

We adopt the system architecture of a typical RTWN, in which multiple sensors and actuators are wirelessly connected to a controller node directly or through relay nodes. 
(Note that the controller node is for initial network setup and performing control computations. \fpas{} does not need it for making any on-line decision and updating schedules.) {\hu We refer to non-controller nodes as device nodes.}
We assume that {\hu all device nodes} have routing capability and are equipped with a single omni-directional antenna to operate on a single channel in half-duplex mode. 
The network is modeled as a directed graph $G = (V, E)$, where the node set $V = \{\{V_0, V_1, \dots\}, V_c \}$ and $V_c$ represents the controller node.
{\extend
A direct link $e = (V_i, V_j)\in E$ represents a wireless link from node $V_i$ to $V_j$ with a Packet Delivery Ratio (PDR) $\lambda^L_e$, which represents the probabilistic transmission success rate on link $e$\footnote{Link PDR $\lambda^L_e$ is usually measured during the site survey and is stable during normal network operations. In case the value of $\lambda^L_e$ changes significantly, the new value is assumed to be broadcast to all the nodes in the network.}. $V_c$ connects to all the nodes via some routes and is responsible for executing relevant control algorithms. 
$V_c$ also contains a network manager which conducts network configuration and resource allocation. 
}

We use the concept of task to describe packet transmission from sensor nodes to actuator nodes. Specifically,
the system runs a fixed set of unicast tasks $\tset = \{\tau_0, \tau_1, \dots, \tau_n\}$. 
Each task $\tau_i \; (0 \leq i \leq n)$ follows a designated single routing path with $H_i$ hops and we use $\overrightarrow{L}_i =[ L_i[0],L_i[1],\dots, L_i[H_i-1]]$ to represent the routing path of $\tau_i$.
It periodically generates a packet which originates at a sensor node, passes through the controller node (not necessary for \fpas{} but to carry out control computations) and delivers a control message to an actuator.
Fig.~\ref{fig:structure} depicts an example RTWN with three tasks running on 7 nodes and task parameters are given in Table~\ref{tab:example}.

\begin{figure}[tb]
  \centering
  \resizebox{1\columnwidth}{!}{\input{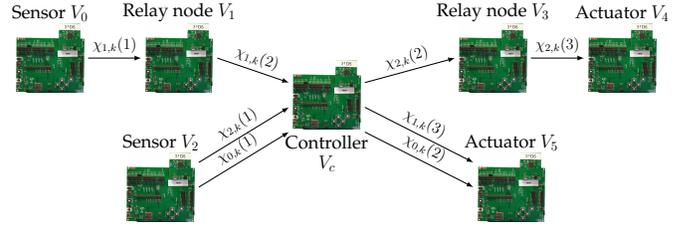}}
  \caption{\small An example RTWN with three unicast tasks.} \label{fig:structure}
  \vspace{-0.1in}
\end{figure}

\begin{table}[tb]
\renewcommand{\arraystretch}{1.3}
\centering
\caption{Task parameters for the example RTWN.}\label{tab:example}
\vspace{-1ex}
\begin{tabular}{|c|c|c|}
\hline
  Task & Routing Path & $P_i$ ($ = D_i$) \\
\hline

$\tau_0$ & $V_2 \rightarrow V_c \rightarrow V_5$ & 9 \\
$\tau_1$ & $V_0 \!\rightarrow\! V_1\! \rightarrow \! V_c \! \rightarrow\! V_5$ & 9 \\
$\tau_2$ & $V_2 \!\rightarrow\! V_c\! \rightarrow \! V_3 \! \rightarrow\! V_4$ & 10\\
\hline
\end{tabular}
\end{table}

\begin{figure}[tb]
  \centering
  {\includegraphics[width=\columnwidth]{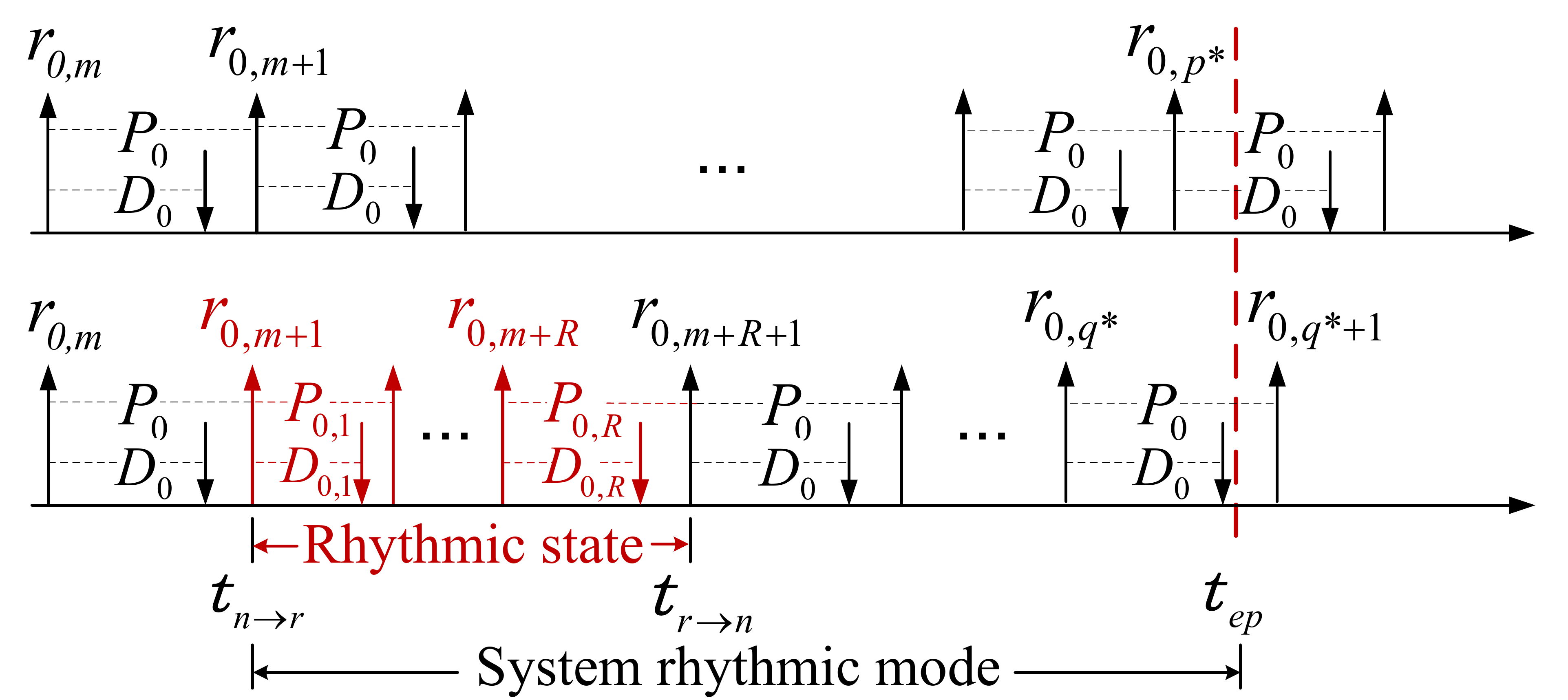}}
  \caption{
  \small Timing parameters of the rhythmic task $\tau_0$ in the system rhythmic mode. Top and bottom subfigures denote the nominal and actual release times and deadlines of $\tau_0$ respectively.} \label{fig:rhythmic}
  \vspace{-0.05in}
\end{figure}

When external disturbances ({\em e.g.}, sudden change in temperature or pressure) occur, many IIoT applications would require more frequent sampling and control actions, which in turn increase network resource demands. 
To capture such abrupt increase in network resource demands, we adopt the rhythmic task model~\cite{Kim_ICCPS12} which has been shown to be effective for handling disturbances in event-triggered control systems~\cite{hong2015online}. (Note that our \fpas{} framework is not limited to the rhythmic task model and is applicable to any task model that provides workload changing patterns for handling disturbances.) In the rhythmic task model,
each unicast task $\tau_i$ has two states: {\em nominal state} and {\em rhythmic state}. In the nominal state, $\tau_i$ follows nominal period $P_i$ and nominal relative deadline $D_i (\leq P_i$), which are all constants.
When a disturbance occurs, $\tau_i$ enters the rhythmic state in which its period and relative deadline are first reduced in order to respond to the disturbance, and then gradually return to their nominal values by following some monotonically non-decreasing pattern. We use vectors $\overrightarrow{P_i}=[P_{i,x},x=1,\dots,R]^{T}$ and $\overrightarrow{D_i}=[D_{i,x},x=1,\dots,R]^{T}$ to represent the periods and relative deadlines of $\tau_i$ when it is in the rhythmic state. 
As soon as $\tau_i$ enters the rhythmic state, its period and relative deadline adopt sequentially the values specified by  $\overrightarrow{P_i}$ and $\overrightarrow{D_i}$, respectively. 
$\tau_i$ returns to the nominal state when it starts using $P_i$ and $D_i$ again. 

Here we assume that at most one task can be in the rhythmic state at any time during the network operation.
To simplify the notation, we refer to any task currently in the \rhystate{} as {\em rhythmic task}
and denote it as $\tau_0$ while task \taui{} is a {\em periodic task} which is currently not in the \rhystate{}. 
As shown in Fig.~\ref{fig:rhythmic}, when $\tau_0$ enters the \rhystate{}, we also say that the system switches to the {\em \rhymode{}}. The system returns to the {\em \normode{}} when the disturbance has been completely handled, typically some time after $\tau_0$ returns to the \norstate{}.
Since disturbances may cause catastrophe to the system, the rhythmic task has a hard deadline when the system is in the rhythmic mode while periodic tasks can tolerate occasional deadline misses.

Each task $\tau_{i}$ consists of an infinite sequence of instances. The $k$-th instance of $\tau_{i}$, referred to as packet $\pkt_{i,k}$, is associated with release time $r_{i,k}$, deadline $d_{i,k}$ and finish time $f_{i,k}$. Without loss of generality, we assume that $\tau_0$ enters the \rhystate{} at $r_{0,m+1}$ (denoted as $t_{n \rightarrow r}$) and returns to the \norstate{} at $r_{0,m+R+1}$ (denoted as $t_{r \rightarrow n}$). Thus,  $\tau_0$ stays in its rhythmic state during $[t_{n \rightarrow r}, t_{r \rightarrow n})$, and $ t_{r \rightarrow n} = t_{n \rightarrow r} + \sum_{x=1}^R{P_{0,x}}$. 
Any packet of $\tau_0$ released in the system \rhymode{} is referred to as a {\em rhythmic packet} while the packets of task \taui{} are {\em periodic packets}.
The delivery of packet $\chi_{i,k}$ at the $h$-th hop is referred to as a transmission denoted as $\chi_{i,k}(h)\ (1 \leq h \leq H_i)$. 

{\extend
Traditionally, RTWNs employ Link-based Scheduling (LBS) to allocate time slots for individual tasks where each slot is allocated to a link by specifying the sender and receiver~\cite{de2014ieee}. If packets from different tasks share a common link and are both buffered at the same sender, their transmission order is decided by a node-specified policy (e.g., FIFO). This approach introduces uncertainty in packet scheduling and may violate the {\han end-to-end (e2e) timing constraints} on packet delivery. To tackle this problem, \emph{Transmission-based Scheduling (TBS)} and \emph{Packet-based Scheduling (PBS)} are proposed in~\cite{zhang2017distributed} and~\cite{flexible}, respectively, to construct deterministic schedules. Each of the two scheduling models has its own advantages and disadvantages and is preferred in different usage scenarios as discussed in~\cite{flexible}. Hence, we consider both  models in our \fpas{} framework. 

In the TBS model, each time slot is allocated to the transmission of a specific packet $\chi_{i,k}$ at a particular hop $h$ 
or kept idle. 
Once the network schedule is constructed, packet transmission in each time slot is unique and fixed.
In the PBS model, each time slot is allocated to a specific packet $\chi_{i,k}$ 
or kept idle. 
Within each time slot assigned to $\chi_{i,k}$, every node along $\chi_{i,k}$'s routing path decides the action to take (e.g., transmit, receive or idle), depending on whether the node has received $\chi_{i,k}$ or not. 
For example, consider a task $\tau_0$ with two slots being assigned in each period. In the TBS model, the first and second slots are dedicated for $\tau_0$'s first and second hops, respectively. In the PBS model, the two slots are allocated to each packet of $\tau_0$ and the second slot can be used to transmit $\tau_0$'s first hop if the transmission fails in the first slot.

Since each link $e$ in the network may suffer packet losses, i.e., $\lambda^L_e <1$, 
packet transmissions may fail, which can significantly affect the timely delivery of real-time packets. To handle such cases, a retransmission mechanism is commonly employed in RTWNs~\cite{song2008wirelesshart,dujovne20146tisch}. Specifically, if a sender node does not receive any ACK from the receiver node within the current slot, it automatically retransmits the packet in the next possible time slot.

To quantify the reliability requirement of the e2e packet delivery for each task, a \emph{required} e2e PDR for $\tau_i$, denoted as $\lambda^R_i$, is introduced. Based on $\lambda^R_i$, the transmission of any packet of $\tau_i$ is reliable if and only if the achieved e2e PDR of $\tau_i$ is larger than or equal to $\lambda^R_i$, i.e., $\lambda_{i,k} \geq \lambda^R_i$. 
To simplify presentation, we assume that all tasks in the network share a common required e2e PDR value, denoted as $\lambda^R$. 
However, our proposed approach can be easily extended to support different $\lambda^R$'s for different tasks.
Table~\ref{tab:notation_table} summarizes the frequently used symbols in this paper.

Based on the above system model, the problem that we aim to solve in this paper is presented as follows.

\noindent
\prob: Assume that for a given RTWN, a static schedule is provided which can guarantee both the e2e timing and reliability requirements of all tasks when there are no disturbances. That is, required number of slots are assigned for each packet (either in the TBS model or PBS model) in the system nominal mode.
Upon detection of a disturbance at $r_{0,m}$ 
(a release time of $\tau_0$'s packet\footnote{We assume that disturbances can be detected only at the time when the sensor samples the environment data, {\em i.e.}, the release time of a certain packet.}), determine the dynamic schedule in the system rhythmic mode such that (i) the system can start handling rhythmic packets no later than $r_{0,m+1}=r_{0,m}+P_0$, (ii) timing and reliability requirements of all the rhythmic packets are satisfied, 
and (iii) the system can safely return to the nominal mode after which all packets can be reliably delivered by their nominal deadlines.
{\tz The objective is to minimize the total reliability degradation on all packets from periodic tasks in the system rhythmic mode.}

Constraint (i) ensures that disturbances can be handled in the earliest possible time ({\em i.e.}, before the nominal arrival time of the next packet). If Constraint (i) were violated, the corresponding control system could become unstable or suffer from severe performance degradation. The meaning of Constraints (ii) and (iii) are self explanatory. 

{\tz 
It has been shown through a motivational example in \cite{fdpas} that centralized packet scheduling approaches ({\em e.g.} OLS and \dpas{}) have two main drawbacks when solving the above problem. First, they rely on a single point ({\em e.g.} the controller) in the network to make on-line decisions for handling the disturbance. This is a significant roadblock in scaling up the packet scheduling framework to be deployed in large-scale RTWNs.
Secondly, centralized approaches suffer from a considerably long response time to the disturbances especially for large RTWNs. This is because centralized approaches require to first send the disturbance information to the controller. After that, a broadcast packet is needed to disseminate the generated dynamic schedule to all nodes in the network to handle the disturbance. In this work, we propose a new approach to address these drawbacks.}
}

\begin{table*}[tb]
 \caption{Summary of important notations {\sh and definitions}}
 \vspace{-0.1in}
 \label{tab:notation_table}
 \centering
{\footnotesize
\begin{tabular}{|c|c|c|c|}
\hline
Notation                                      & Definition                                                & Notation                        & Definition                                                                \\ \hline
$V_{j}$ ($j = 0,1,\ldots$), $V_c$             & Device nodes and controller node                         & \tnr,                           & Slot when $\tau_{0}$ leaves its nominal state                               \\ \cline{1-2}
$\tau_i (0 \leq i \leq n)$                    & Unicast tasks                                            & \trn                            & and its rhythmic state, respectively                      \\ \hline
$H_i$                                         & Number of hops of $\tau_{i}$                             & $t_{sp}$, $t_{ep}$,             & Start point, end point,  end point candidate                  \\ \cline{1-2}
$P_{i}$ ($D_{i}$)                              & Nominal period (deadline) of $\tau_i$                    & $t_{ep}^c$, $t_{ep}^u$          & and end point upper bound                \\ \hline
{\rule{0pt}{1.35em}$\overrightarrow{P_i}$} ($\overrightarrow{D_i}$) & Rhythmic period (deadline) vector of $\tau_i$ & $S$, $\dynamic$             & Static schedule and dynamic schedule                                       \\ \hline
$\pkt_{i,k}$                                & The $k$-th released packet of task $\tau_i$               & \dnodes{}                      & Set of nodes receiving the disturbance information                                               \\ \hline
$\pkt_{i,k}(h)$                               & The {\em h}-th transmission of packet $\pkt_{i,k}$      & $\Psi(t)$, $\Phi(t)$              & Set of active rhythmic and periodic packets                            \\ \hline
$\lambda^R$              & Required e2e packet delivery ratio (for all tasks)       &  $\dropset[t_{sp},t)$               & Set of dropped periodic packets                                       \\ \cline{1-2}
$\lambda_{i,j}$, $\overrightarrow{R}_{i,j}$   & E2e PDR value and retry vector of $\chi_{i,j}$      & $\dropset^*[t_{sp},t)$                  & and transmissions within $[t_{sp},t)$                                  \\ \hline
$R_{i,j}[h]$           & Number of trials for $h$-th hop assigned by $\overrightarrow{R}_{i,j}$          & $\delta_{i,j}$                     & PDR degradation of $\chi_{i,j}$          \\ \hline
\end{tabular}
}
\vspace{-0.1in}
\end{table*}

\section{Overall Framework of \fpas{}}
\label{sec:framework}
\eat{
In this section, we first give a motivating example to show the deficiency of centralized packet scheduling approaches for handling rhythmic tasks. We then present an overview of the fully distributed packet scheduling framework, \fpas{}.

\subsection{Drawbacks of Centralized Approaches}\label{ssec:motivation}


In order to properly handle unexpected external disturbances, centralized dynamic scheduling approaches have been proposed, which can adapt to changes in on-line network resource demand. For example, OLS~\cite{hong2015online} can generate an on-line dynamic schedule based on a dynamic programming approach to handling disturbances modeled as rhythmic events. To further improve the performance in terms of {\extend reducing the degradation on timing requirements of periodic packets},\footnote{Both OLS and \dpas{} assume all links are reliable in the network and only consider the timing requirements of all tasks.} \dpas{}~\cite{zhang2017distributed} leverages the network-wide synchronization in RTWNs and the computing capability of individual device nodes to generate consistent schedules locally. By effectively reducing the amount of schedule related information to be broadcast by the gateway, \dpas{} significantly improves the scalability of the dynamic schedule construction and dissemination processes.  

Centralized approaches, however, incur long latency for handling disturbances, especially in large  networks. Consider an RTWN (Fig.~\ref{fig:structure}) with 3 tasks ($\tau_0, \tau_1$ and $\tau_2$) running on 7 nodes ($V_0, \ldots, V_5$ and $V_c$) with $V_0$ and $V_2$ being sensors, $V_4$ and $V_5$ being actuators, $V_1$ and $V_3$ being relay nodes, and $V_c$ being the controller node and functioning as the gateway in centralized approaches. 
Note that since centralized approaches rely on the controller node to disseminate the dynamic schedule, broadcast task $\tau_3$ is needed.
The tasks' routing paths, periods and relative deadlines 
are given in \pref{tab:example}. 

\begin{figure}[tb]
  \centering
  \resizebox{1\columnwidth}{!}{\input{Figures/example.tikz}}
  \caption{\small An example RTWN with three unicast tasks and one broadcast task running on 7 nodes.} \label{fig:structure}
\end{figure}

\eat{
\begin{table}[b]
\renewcommand{\arraystretch}{1.3}
\centering
\caption{\small Task parameters for the motivational example}\label{tab:example}
\begin{tabular}{|m{0.56cm}<{\centering}|m{2.5cm}<{\centering}|m{1.2cm}<{\centering}|m{.56cm}<{\centering}|m{.8cm}<{\centering}|m{0.8cm}<{\centering}|}
\hline
  Task & Routing Path & $P_i$ ($ = D_i$) & $\overrightarrow{P_i}$ & \mbox{$\overrightarrow{D_i}$} \\
\hline

$\tau_0$ & $V_2 \rightarrow V_g \rightarrow V_5$ & 9 & $[4, 6]^T$ & $[3, 5]^T$ \\
$\tau_1$ & $V_0 \!\rightarrow\! V_1\! \rightarrow \! V_g \! \rightarrow\! V_5$ & 9 & N/A & N/A \\
$\tau_2$ & $V_2 \!\rightarrow\! V_g\! \rightarrow \! V_3 \! \rightarrow\! V_4$ & 10 & N/A & N/A \\
$\tau_3$ & $V_g \rightarrow *$ & 18 & N/A & N/A \\
\hline
\end{tabular}
\end{table}
}

\begin{table}[tb]
\renewcommand{\arraystretch}{1.3}
\centering
\caption{Task parameters for the motivational example}\label{tab:example}
\vspace{-1ex}
\begin{tabular}{|m{0.56cm}<{\centering}|m{2.5cm}<{\centering}|m{1.2cm}<{\centering}|m{.56cm}<{\centering}|}
\hline
  Task & Routing Path & $P_i$ ($ = D_i$) \\
\hline

$\tau_0$ & $V_2 \rightarrow V_c \rightarrow V_5$ & 9 \\
$\tau_1$ & $V_0 \!\rightarrow\! V_1\! \rightarrow \! V_c \! \rightarrow\! V_5$ & 9 \\
$\tau_2$ & $V_2 \!\rightarrow\! V_c\! \rightarrow \! V_3 \! \rightarrow\! V_4$ & 10\\
$\tau_3$ & $V_c \rightarrow *$ & 18\\
\hline
\end{tabular}
\end{table}

Assume all tasks are synchronized and first released at time slot $0$, and each node employs an EDF scheduler to construct its local schedule (see Fig.~\ref{fig:example}). Suppose at time slot $9$, an external disturbance is detected and sensor $V_2$ sends a rhythmic event request via $\tau_0$ to the controller node. $V_c$ then determines the time slot $\tnr$ when $\tau_0$ is going to enter its rhythmic state. In order to achieve fast response to the disturbance, $\tnr$ should be set to be as early as possible, but later than the time slot when all nodes in the network receive the dynamic schedule. 
In this example, $V_c$ has to wait till time slot $26$ to broadcast the constructed dynamic schedule. Only after the broadcast packet reaches all nodes at $30$, $\tau_0$ can enter its rhythmic state at the nearest release time slot $36$. Therefore, for this example, though the disturbance is detected by the sensor at time slot $9$, the system cannot enter the rhythmic mode until slot $36$, which is three nominal periods later (instead of one nominal period as required in \prob).

From the above example, one can readily see that the centralized approaches suffer from a considerably long response time to the disturbances especially for large RTWNs. Moreover, centralized approaches rely on a single point (the gateway) in the network to make on-line packet scheduling decisions. These are the two main roadblocks in scaling up the packet scheduling framework to support handling disturbances in large-scale RTWNs. 

\eat{
when a disturbance is detected, the sensor sends a rhythmic event request via the packet of $\tau_0$ (referred to as $\pkt_{0,m}$). $t'$ denotes the time slot that the gateway receives the disturbance from $\pkt_{0,m}$. 
$V_g$ generates the dynamic schedule for handling $\tau_0$ before the upcoming broadcast packet $\pkt_{n+1,k}$ after $t'$, then, propagates to the network via $\pkt_{n+1,k}$. Suppose $t''$ is the time slot that $\pkt_{n+1,k}$ reaches all nodes in the network. $\tau_0$ enters its rhythmic state at $\tnr$, {\em i.e.}, its first release time after $t''$. That is, $\tau_0$ raises the rhythmic event request at $r_{0,m}$ while enters the rhythmic state till $\tnr$ to wait at least one broadcast packet being transmitted to the network. \footnote{In OLS, if a feasible dynamic schedule cannot be generated by the dynamic programming based approach before the upcoming broadcast packet, it has to wait till the next broadcast packet.} Therefore, the response time to the rhythmic event request of both OLS and \dpas{} can be extremely long especially in a large RTWN since a broadcast packet with longer routing path is necessary.

Besides, centralized or hybrid approaches both rely on a single point in the network to make packet scheduling decisions. This is the main difficulty in scaling up to large RTWNs. 
All the drawbacks above will cause performance degradation to the system which motivates us to design a better framework in terms of quick response for dealing with disturbances and scalability of RTWNs.

\begin{figure}[tb]
  \centering
  {\includegraphics[width=\linewidth]{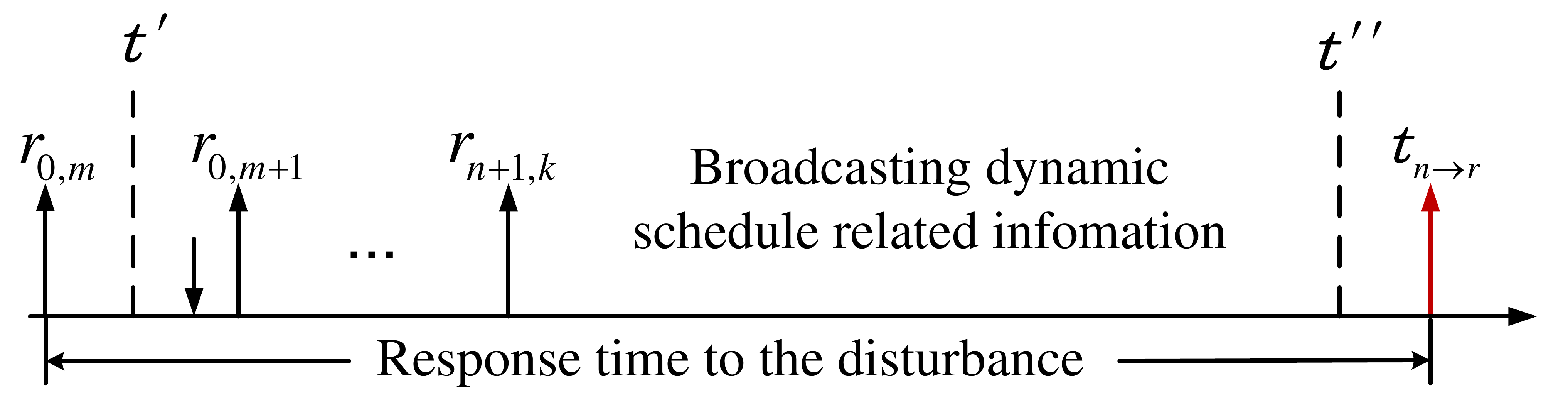}}
  \caption{
  \small Overall frameworks of OLS and \dpas{}.} \label{fig:motivation} 
\end{figure}
}
}

In order to achieve fast response to disturbances in RTWNs, in this work we propose a fully distributed packet scheduling framework, referred to as \fpas{}. The key idea of \fpas{} is to make dynamic, local schedule adaptation at each node along the path of the rhythmic task while avoiding transmission collisions from other nodes that still follow their static schedules in the system rhythmic mode.

\eat{
\begin{figure}[tb]
  \centering
  {\includegraphics[width=\linewidth]{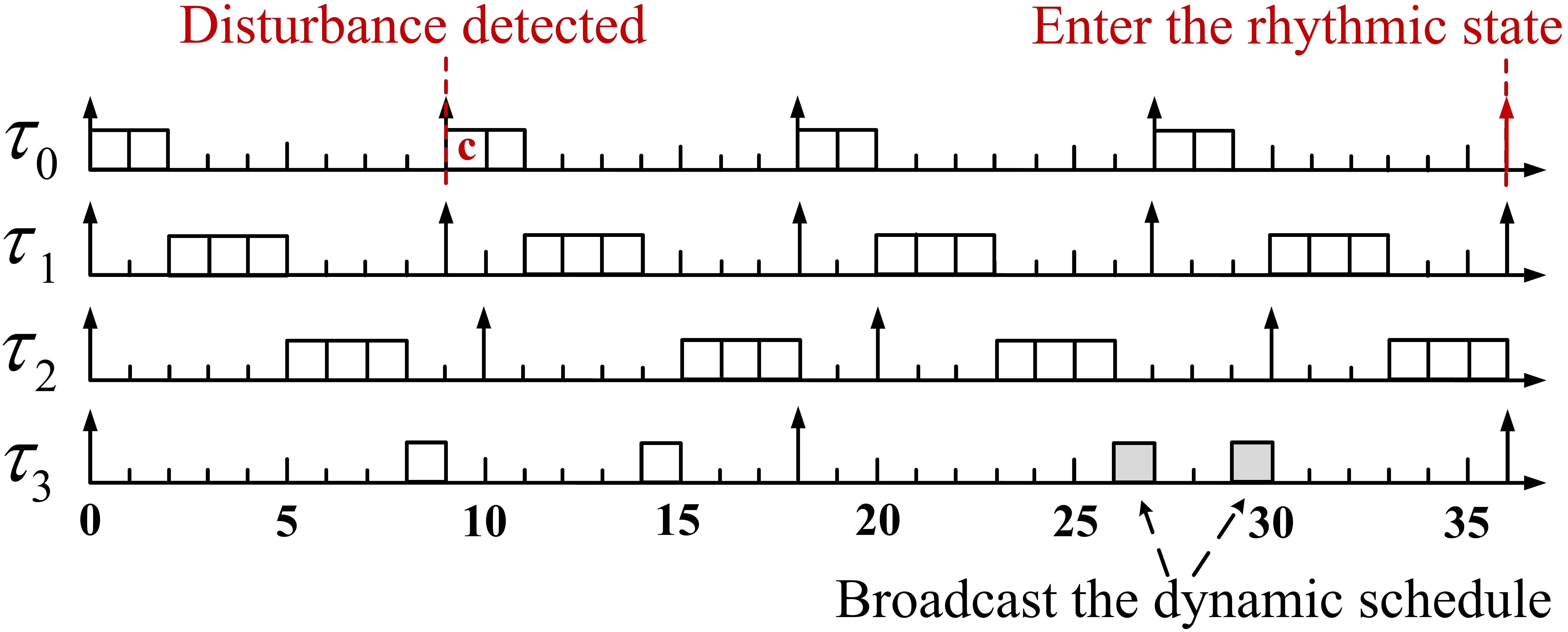}}
  \caption{\small Local EDF schedules of the tasks in the motivating example. The block with symbol $c$ denotes the transmission of the rhythmic event request. The shaded blocks denote the two transmissions of the broadcast task to propagate the dynamic schedule generated at the controller node to the whole network.} \label{fig:example}
\end{figure}
}

Fig.~\ref{fig:overview} gives an overview of the execution model of \fpas{}. After network initialization, each node generates locally a static schedule, $S$, using the local schedule generation mechanism in \dpas{} and follows $S$ to transmit packets.  When a disturbance is detected by rhythmic task $\tau_0$ at $t'=r_{0,m}$, a notification is propagated to all the nodes responsible for handling the disturbance. 
Let these nodes be $V_j \in$ \dnodes{}. Upon receiving the notification, each node in \dnodes{} determines the time duration of the network being in the rhythmic mode and generates a dynamic schedule \dynamic{} for handling the disturbance. Starting from $r_{0,m+1}$, one nominal period of $\tau_0$ after detecting the disturbance, the nodes in \dnodes{} follow \dynamic{} while all other nodes keep using static schedule $S$ to transmit periodic packets.
Thus, by not relying on a broadcast packet to disseminate the dynamic schedule generated by a centralized point in the network, \fpas{} is able to significantly reduce the response time of reacting to disturbances. For ease of discussion, in the rest of the paper, we refer to {\em disturbance response time} (DRT) as the time duration from $t'$ to the start time of system rhythmic mode and {\em disturbance handling latency} (DHL) as the time duration of system rhythmic mode (see Fig.~\ref{fig:overview}).

To ensure that \fpas{} works properly, several challenges need to be tackled. First, when a disturbance occurs, only the sensor node that has detected it knows which task will enter the rhythmic state, while the rest of the nodes in \dnodes{} that are to handle the disturbance have no knowledge about this. Second, if the nodes in \dnodes{} follow the dynamic schedule while other nodes follow the static schedule $S$, transmission collisions would occur which may cause rhythmic packets to {\extend violate their timing and reliability requirements ({\em e.g.} missing deadlines)}. Third, to properly handle disturbances, efficient methods are needed by the nodes in \dnodes{} to determine a dynamic schedule in which {\extend the reliability degradation on periodic packets is minimized}. We discuss in detail how \fpas{} tackles these challenges in the following sections.

\eat{
\begin{algorithm}[tb]
\caption{Main function of \fpas{}}
 \label{alg:overview}
\begin{algorithmic}[1]
\small
\WHILE {true}
\STATE {Every node generates and follows the static schedule locally.}
\IF {a disturbance is detected by the sensor $N_j$}
    \STATE{$N_j$ sends the disturbance to the node set \dnodes{}.}
    \STATE{Each node in \dnodes{} generates the dynamic schedule to follow.}
    \STATE{The rhythmic task enters its rhythmic state.}
\ENDIF
\ENDWHILE
\end{algorithmic}
\end{algorithm}
}

\eat{
To achieve this, there are several challenges to be tackled. 
\begin{itemize}
    \item When an unexpected disturbance occurs, how to let the necessary local nodes to know the disturbance.
    \item Since some nodes know the disturbance while others not, how to avoid transmission collisions occurred among these nodes.
    \item How to drop the minimum number of periodic packets to guarantee the deadlines of rhythmic packets.
\end{itemize}

In the following sections, we answer the three questions one by one.
}

\section{Propagating Disturbance Information}
\label{sec:propagating}
In centralized approaches, all nodes in the RTWN must know the disturbance information since a dynamic schedule must be deployed at each node. However, such a network-wide propagation mechanism does not scale and often violates constraint (i) in \prob{} as shown by the motivating example. To overcome this drawback, we propose to disseminate the disturbance information to only a subset of all nodes, denoted as \dnodes{}, to minimize the DRT. This scheme requires the following three questions be answered: (1) which nodes in the network belong to \dnodes{}, (2) how to propagate the disturbance information to nodes in \dnodes{}, and (3) does each node in \dnodes{} have sufficient time to generate the dynamic schedule before the system enters the rhythmic mode? Below we present our answers to these questions.

\begin{figure}[tb]
  \centering
  {\includegraphics[width=\linewidth]{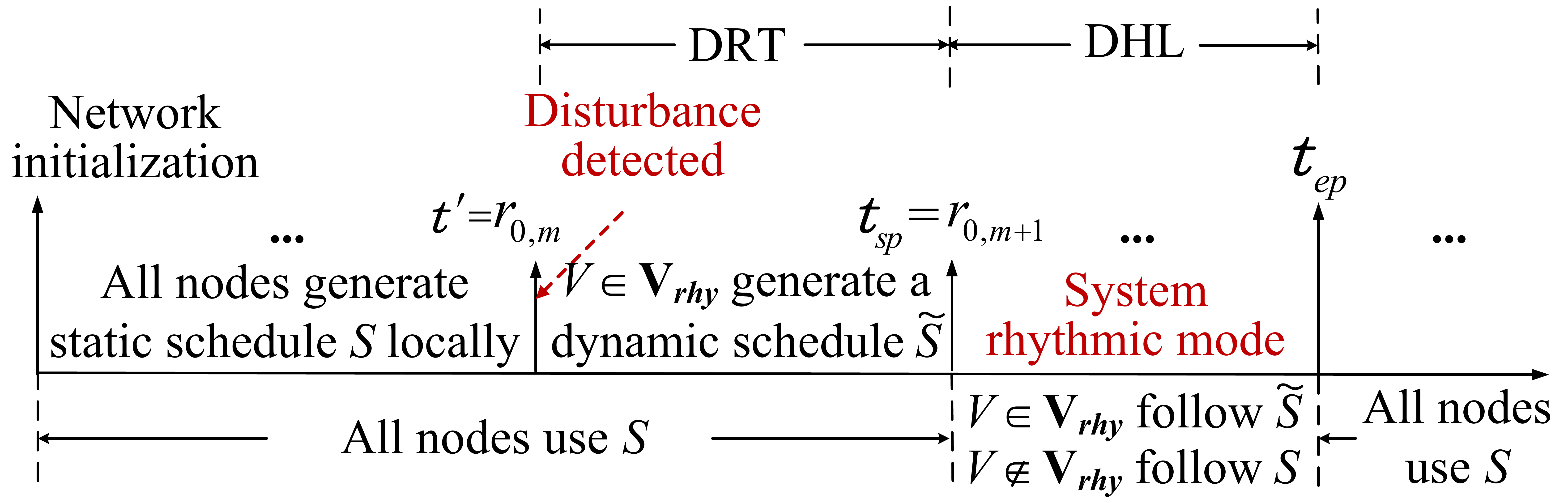}}
  \caption{\small Overview of the execution model of \fpas{}.} \label{fig:overview}
\end{figure}

Consider questions (1) and (2) above. Recall that when a disturbance occurs, the rhythmic task $\tau_0$ will enter its rhythmic state following reduced periods and deadlines as specified in $\overrightarrow{P_0}$ and $\overrightarrow{D_0}$. An updated schedule is needed to accommodate the increased workload of $\tau_0$. To ensure that each (re)transmission $\pkt_{0,k}(h)$ can be successful, both the sender and the receiver of $\pkt_{0,k}(h)$ must follow the same schedule. Thus, all nodes along the routing path of $\tau_0$ must know the disturbance information to generate a consistent dynamic schedule, and should be included in \dnodes{}. For example, \dnodes{} $= \{V_2, V_c, V_5\}$ for the example in Fig.~\ref{fig:structure} when $\tau_0$ enters the rhythmic state. When a disturbance is detected at $r_{0,m}$, its information can be piggybacked onto $\pkt_{0,m}$ and transmitted to all nodes in \dnodes{}. Propagating disturbance information in this manner guarantees that all nodes in \dnodes{} receive the disturbance information within one nominal period of $\tau_0$, {\em i.e.}, $P_0$, since the static schedule ensures that each task {\han is assigned with the required} number of transmission and retransmission slots along its routing path within $P_0$ in order to {\extend meet the e2e timing and reliability requirements.}

Now consider question (3). As required in Constraint (i) of {\prob}, the system should start handling the rhythmic packets from $r_{0,m+1}$ after the disturbance is detected at $r_{0,m}$. This requires that (i) 
the disturbance information be successfully propagated to the relevant nodes before
$\tau_0$ enters its rhythmic state at $r_{0,m+1}$, and (ii) each node in \dnodes{} completes the construction of the dynamic schedule before it starts receiving/transmitting the first rhythmic packet. The propagation scheme discussed above ensures that condition (i) is met. Regarding condition (ii), our prior work showed that one idle slot (10ms) is sufficient for a typical device node in RTWNs ({\em e.g.}, TI CC2538 SoC) to complete its local schedule computation~\cite{zhang2017distributed}. The theorem below establishes that such an idle slot indeed exists within the time frame specified in condition (ii). 

\begin{theorem}\label{thm:idle}
 If an RTWN system is schedulable under a given static schedule, any node $V_j$ ($V_j \neq V_c$) in \dnodes{} has at least one idle slot (neither receiving nor sending any transmission) between time $t_1$ ($t_1 \geq r_{0,m}$) when it receives the disturbance information and 
 time $t_2$ ($t_2 \geq r_{0,m+1}$) when it is involved in the transmission of the first rhythmic packet after $\tau_0$ enters its rhythmic state at $r_{0,m+1}$.
\end{theorem}

\begin{IEEEproof}
We first recall the following lemma from~\cite{zhang2017distributed}.

\begin{lemma}\label{lem:idle}
If an RTWN system is schedulable under a given static schedule, {\em i.e.} each packet completes all its transmissions before the deadline, for any node $V_j \neq V_c$ and task $\tau_i$ passing through $V_j$, there exists at least one idle slot at $V_j$ among any three consecutive transmissions of $\tau_i$ passing $V_j$.
\end{lemma}

Since in our system model, sensors and actuators are connected via the controller node, every task follows a routing path with at least two hops corresponding to two transmissions {\extend (assigned with multiple transmission and retransmission slots)}. Suppose $\pkt_{0,m}(h)$ occurs at $t_1$ and is the transmission from which $V_j$ receives the disturbance information\footnote{If $V_j$ is the sensor, it detects the disturbance at $r_{0,m}$.}. There exists at least one transmission between $\pkt_{0,m}(h)$ and $\pkt_{0,m+1}(h)$
(the first transmission that $V_j$ is involved in the dynamic schedule, occurring at $t_2$). 
Then, according to Lemma~\ref{lem:idle}, $V_j$ has at least one idle slot between $\pkt_{0,m}(h)$ and $\pkt_{0,m+1}(h)$ ({\em i.e.}, between $t_1$ and $t_2$). Thus, the theorem holds.
\end{IEEEproof}

Based on Theorem~\ref{thm:idle} and the disturbance propagation time bound, the proposed partial disturbance propagation scheme guarantees that any disturbance can be promptly responded within one nominal period of the rhythmic task and Constraint (i) in {\prob} can be satisfied.

\section{Avoiding Transmission Collisions}
\label{sec:collisions}

According to the disturbance propagation mechanism presented in Section~\ref{sec:propagating}, only the nodes on the path of the rhythmic task are included in \dnodes{}. Nodes in \dnodes{} construct their local schedules individually and employ them in the system rhythmic mode. All other nodes in the network follow the static schedule. With this execution model, unless the disturbance information is propagated to the entire RTWN, inconsistencies between the dynamic and static schedules in the system rhythmic mode may easily arise, which would result in transmission collisions. To ensure that the disturbances are handled appropriately, in the \fpas{} framework, the transmissions of rhythmic packets need to be always successful even in the presence of collision with other periodic packets.

In conventional RTWNs such as WirelessHART~\cite{song2008wirelesshart} and 6TiSCH~\cite{dujovne20146tisch}, TDMA-based data link layer are widely adopted to provide synchronized and collision-free channel access. In addition, most of those protocols employ the Clear Channel Assessment (CCA) operation at the beginning of each transmission for collision avoidance. CCA, however, cannot prioritize packet transmissions. When multiple transmissions happen in the same time slot sharing the same destination, it cannot guarantee the more important packets ({\em e.g.}, rhythmic packets) are granted the access to the channel. 

\eat{In IEEE 802.15.4e standards, a timeslot is allocated for a single packet transmission from the sender to the receiver. Within a timeslot, the sender transmits the packet, while the receiver acknowledges the reception of the packet if it is valid. A transmission is considered successful only if the sender receives the correct acknowledgement. If multiple transmissions happen in one timeslot within the reception range, and the strongest signal can not exceed rated Signal-to-Interference-plus-Noise Ratio (SINR), all of them will fail due to the  interference. Thus, in a normal network setup, a proper schedule is maintained so that no such interference will happen.

When a dynamic schedule for handling the disturbance is generated by the necessary nodes (the nodes involved in rhythmic task), the challenge is how to avoid transmission collisions between that of the dynamic schedule running on the necessary nodes and the original schedule running on other nodes.}

To tackle this challenge, we propose an enhancement to the IEEE 802.15.4e standard~\cite{de2014ieee}, called Multi-Priority MAC (\pmac{}), to support prioritization of packet transmissions in RTWNs. Several attempts have been made in the literature towards supporting this feature. 
For example, the PriorityMAC was proposed in~\cite{prioritymac} to prioritize critical traffic in RTWNs. It introduces the concept of subslots, in which the transmitter does a very short transmission to indicate the priority of the packet to be transmitted in the following time slot. By adding two subslots {\hu before} each time slot, PriorityMAC is able to create three priority levels. Different from PriorityMAC, the design of the {\pmac{}} aims to be lightweight and scalable. In {\pmac{}}, the transmitter does not explicitly conduct a short transmission to indicate the priority. Instead it implicitly indicates the priority of the transmission by adjusting the Start-Of-Frame (SOF) time offset. Compared with PriorityMAC, {\pmac{}} is more energy efficient (by avoiding transmissions in the subslots), and able to support more priority levels.

Fig.~\ref{fig:timeslot} gives a comparison of the slot timing of 802.15.4e (top) and {\pmac{} (bottom)}. In a 802.15.4e time slot, the sender transmits a packet and the receiver responds with an acknowledgement (ACK) if the packet is successfully received\footnote{No acknowledgement is provided for broadcast and multicast packets.}. The packet transmission starts at \emph{TxOffset} after the start of the time slot, while the ACK starts at \emph{TxAckDelay} after the completion of the packet transmission. A long Guard Time (LGT) and a short Guard Time (SGT) are used by the receiver and sender respectively to tolerate clock drift and radio/CPU operation delays. With this standard design of 802.15.4e, if multiple senders transmit packets in the same time slot, they are not aware of the other transmissions, and thus will cause interference.  The slot timing of {\pmac{}} is presented at the bottom of Fig.~\ref{fig:timeslot}. In {\pmac{}}, instead of being set as a constant, \emph{TxOffset} is varied to implicitly indicate the priority of the packet (shown as red dashed lines). A packet with a higher priority is associated with a shorter \emph{TxOffset} to start the transmission earlier. In addition, a CCA operation will be performed before each transmission to ensure that there is no higher priority packet transmission present in the channel. This enhancement ensures that only the highest priority packet (with the shortest \emph{TxOffset}) is transmitted, and all lower priority transmissions yield to it. \eat{Note that the level of priorities that can be supported by {\pmac{}} is not limited to 4, but can be adjusted by adding or removing \emph{TxOffset}s.}

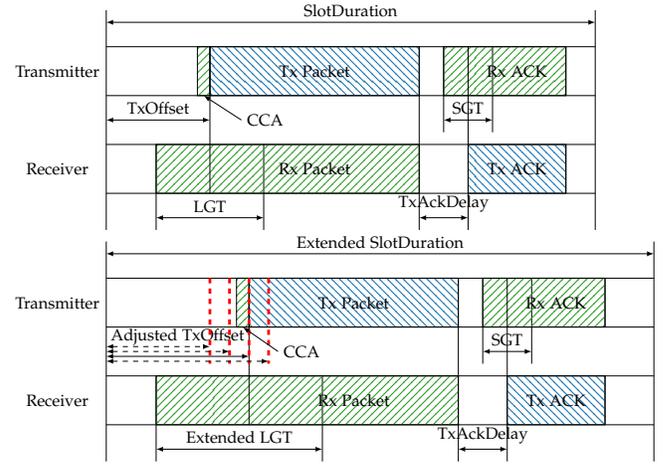
\begin{figure}[t]
\centering
\resizebox{1\columnwidth}{!}{\begin{tikzpicture}
\definecolor{TXColor}{RGB}{31,120,180}%
\definecolor{RXColor}{RGB}{51,160,44}%
\tikzset{TXBlock/.style={pattern=north west lines, pattern color=TXColor}}
\tikzset{RXBlock/.style={pattern=north east lines, pattern color=RXColor}}%
\node at (11.5,0) {};
\def\SlotDuration{10}%
\def\TxOffset{2.12}%
\pgfmathsetmacro\TxCCAOffset{\TxOffset-0.256}%
\pgfmathsetmacro\TxCCAText{(\TxCCAOffset +\TxOffset)/2}%
\def\TxAckDelay{1}%
\def\LongGT{1.1}%
\def\ShortGT{0.5}%
\pgfmathsetmacro\RxOffset{\TxOffset-\LongGT}%
\pgfmathsetmacro\RxStop{\TxOffset+\LongGT}%
\def\TxEnd{6.4}%
\pgfmathsetmacro\TxAckOffset{\TxEnd+\TxAckDelay}%
\pgfmathsetmacro\RxAckOffset{\TxAckOffset-\ShortGT}%
\pgfmathsetmacro\RxAckEnd{\TxAckOffset+\ShortGT}%
\pgfmathsetmacro\TxAckEnd{\TxAckOffset+2}%
\pgfmathsetmacro\TxText{(\TxOffset + \TxEnd) / 2}%
\pgfmathsetmacro\TxAckText{(\TxAckOffset + \TxAckEnd) / 2}%
\node at (-1,1.5) {Transmitter};
\node at (-1,-0.5) {Receiver};
\draw  (0,2.75) -- (0,-1.75);
\draw  (\SlotDuration,2.75) -- (\SlotDuration,-1.75);
\draw  (0,2) -- (\SlotDuration,2);
\draw  (\SlotDuration,1) -- (0,1);
\draw  (0,0) -- (\SlotDuration,0);
\draw  (\SlotDuration,-1) -- (0,-1);
\draw (\TxEnd,2) -- (\TxEnd,-1.75);
\draw [TXBlock] (\TxOffset,2) rectangle (\TxEnd,1);
\draw [RXBlock] (\TxCCAOffset,2) rectangle (\TxOffset,1);
\node at (\TxText,1.5) {Tx Packet};
\draw (\RxOffset,0) -- (\RxOffset,-1.75);
\draw (\RxStop,0) -- (\RxStop,-1.75);
\draw [RXBlock] (\RxOffset,0) rectangle (\TxEnd,-1);
\node at (\TxText,-0.5) {Rx Packet};
\draw (\TxOffset,2) -- (\TxOffset,-1);
\draw (\RxAckOffset,2) -- (\RxAckOffset,0.25);
\draw (\TxAckEnd,2) -- (\TxAckEnd,1);
\draw (\RxAckEnd,2) -- (\RxAckEnd,0.25);
\draw [RXBlock] (\RxAckOffset,2) rectangle (\TxAckEnd,1);
\node at (\TxAckText,1.5) {Rx ACK};
\draw (\TxAckOffset,2) -- (\TxAckOffset,-1.75);
\draw (\TxAckEnd,0) -- (\TxAckEnd,-1);
\draw [TXBlock] (\TxAckOffset,0) rectangle (\TxAckEnd,-1);
\node at (\TxAckText,-0.5) {Tx ACK};
\draw [arrows={latex-latex}] (0,2.5) -- node [above,midway] {SlotDuration} (\SlotDuration,2.5);
\draw [arrows={latex-latex}] (0,0.5) -- node [above,midway] {TxOffset} (\TxOffset,0.5);
\draw [arrows={latex-latex}] (\TxEnd,-1.5) -- node [above,midway] {TxAckDelay} (\TxAckOffset,-1.5);
\draw [arrows={latex-latex}] (\RxAckOffset,0.5) -- node [above,midway] {SGT} (\RxAckEnd,0.5);
\draw [arrows={latex-latex}] (\RxOffset,-1.5) -- node [above,midway] {LGT} (\RxStop,-1.5);
\node (cca) at (3.25,0.5) {CCA};
\draw [arrows={latex-}] (\TxCCAText,1) -- (cca.west);
\end{tikzpicture}}
\resizebox{1\columnwidth}{!}{\begin{tikzpicture}
\definecolor{TXColor}{RGB}{31,120,180}%
\definecolor{RXColor}{RGB}{51,160,44}%
\tikzset{TXBlock/.style={pattern=north west lines, pattern color=TXColor}}
\tikzset{RXBlock/.style={pattern=north east lines, pattern color=RXColor}}%
\node at (11.5,0) {};
\def\PriorityTick{0.4}%
\pgfmathsetmacro\SlotDuration{10 + \PriorityTick * 3}%
\pgfmathsetmacro\TxOffsetZero{2.12 + \PriorityTick * 0}%
\pgfmathsetmacro\TxOffsetOne{2.12 + \PriorityTick * 1}%
\pgfmathsetmacro\TxOffsetTwo{2.12 + \PriorityTick * 2}%
\pgfmathsetmacro\TxOffsetThree{2.12 + \PriorityTick * 3}%
\pgfmathsetmacro\TxOffset{\TxOffsetTwo}%
\pgfmathsetmacro\TxCCAOffset{\TxOffset-0.256}%
\pgfmathsetmacro\TxCCAText{(\TxCCAOffset +\TxOffset)/2}%
\def\TxAckDelay{1}%
\def\LongGT{1.1}%
\def\ShortGT{0.5}%
\pgfmathsetmacro\RxOffset{\TxOffsetZero-\LongGT}%
\pgfmathsetmacro\RxStop{\TxOffsetThree+\LongGT}%
\pgfmathsetmacro\TxEnd{\TxOffset+4.28}%
\pgfmathsetmacro\TxAckOffset{\TxEnd+\TxAckDelay}%
\pgfmathsetmacro\RxAckOffset{\TxAckOffset-\ShortGT}%
\pgfmathsetmacro\RxAckEnd{\TxAckOffset+\ShortGT}%
\pgfmathsetmacro\TxAckEnd{\TxAckOffset+2}%
\pgfmathsetmacro\TxText{(\TxOffset + \TxEnd) / 2}%
\pgfmathsetmacro\TxAckText{(\TxAckOffset + \TxAckEnd) / 2}%
\node at (-1,1.5) {Transmitter};
\node at (-1,-0.5) {Receiver};
\draw  (0,2.75) -- (0,-1.75);
\draw  (\SlotDuration,2.75) -- (\SlotDuration,-1.75);
\draw  (0,2) -- (\SlotDuration,2);
\draw  (\SlotDuration,1) -- (0,1);
\draw  (0,0) -- (\SlotDuration,0);
\draw  (\SlotDuration,-1) -- (0,-1);
\draw (\TxEnd,2) -- (\TxEnd,-1.75);
\draw [TXBlock] (\TxOffset,2) rectangle (\TxEnd,1);
\draw [RXBlock] (\TxCCAOffset,2) rectangle (\TxOffset,1);
\node at (\TxText,1.5) {Tx Packet};
\draw (\RxOffset,0) -- (\RxOffset,-1.75);
\draw (\RxStop,0) node (v1) {} -- (\RxStop,-1.75);
\draw [RXBlock] (\RxOffset,0) rectangle (\TxEnd,-1);
\node at (\TxText,-0.5) {Rx Packet};
\draw (\TxOffset,2) -- (\TxOffset,-1);
\draw (\RxAckOffset,2) -- (\RxAckOffset,0.25);
\draw (\TxAckEnd,2) -- (\TxAckEnd,1);
\draw (\RxAckEnd,2) -- (\RxAckEnd,0.25);
\draw [RXBlock] (\RxAckOffset,2) rectangle (\TxAckEnd,1);
\node at (\TxAckText,1.5) {Rx ACK};
\draw (\TxAckOffset,2) -- (\TxAckOffset,-1.75);
\draw (\TxAckEnd,0) -- (\TxAckEnd,-1);
\draw [TXBlock] (\TxAckOffset,0) rectangle (\TxAckEnd,-1);
\node at (\TxAckText,-0.5) {Tx ACK};
\draw [arrows={latex-latex}] (0,2.5) -- node [above,midway] {Extended SlotDuration} (\SlotDuration,2.5);
\draw [arrows={latex-latex}] (0,0.4) -- node [above=3pt,midway] {Adjusted TxOffset} (\TxOffset,0.4);
\draw [arrows={latex-latex}] (\TxEnd,-1.5) -- node [above,midway] {TxAckDelay} (\TxAckOffset,-1.5);
\draw [arrows={latex-latex}] (\RxAckOffset,0.5) -- node [above,midway] {SGT} (\RxAckEnd,0.5);
\draw [arrows={latex-latex}] (\RxOffset,-1.5) -- node [above,midway] {Extended LGT} (\RxStop,-1.5);
\draw [dashed,ultra thick,color=red] (\TxOffsetZero,2) -- (\TxOffsetZero,0.25);
\draw [dashed,ultra thick,color=red] (\TxOffsetOne,2) -- (\TxOffsetOne,0.25);
\draw [dashed,ultra thick,color=red] (\TxOffsetTwo,2) -- (\TxOffsetTwo,0.25);
\draw [dashed,ultra thick,color=red] (\TxOffsetThree,2) -- (\TxOffsetThree,0.25);
\draw [dashed,arrows={latex-latex}] (0,0.6) -- (\TxOffsetZero,0.6);
\draw [dashed,arrows={latex-latex}] (0,0.5) -- (\TxOffsetOne,0.5);
\draw [dashed,arrows={latex-latex}] (0,0.3) -- (\TxOffsetThree,0.3);
\node (cca) at (4,0.5) {CCA};
\draw [arrows={latex-}] (\TxCCAText,1) -- (cca.west);
 \end{tikzpicture}}
\caption{\small Slot timing of 802.15.4e (top) and MP-MAC (bottom)}
\label{fig:timeslot}
\end{figure}

Similar to the guard times, the \emph{TxOffset} values for different priorities need to be set sufficiently apart so that different senders and receivers have consensus on the priorities. In {\pmac{}}, we define \emph{PriorityTick} as the difference between two consecutive \emph{TxOffset}s. To support $k$ different priorities in {\pmac{}}, the length of the time slot, compared to the standard design, needs be extended by $(k-1) \times PriorityTick$. A longer \emph{PriorityTick} can ensure  successful packet prioritization, but either leads to longer SlotDuration and reduced network throughput, or smaller number of supported priorities if the size of the time slot is fixed. Since \emph{PriorityTick} is a hardware-dependent parameter, we will elaborate the selection of {\em PriorityTick} in our testbed experiments and demonstrate the effectiveness of {\pmac{}} in Section~\ref{sec:testbed}.

\section{System Rhythmic Mode}
\label{sec:dropping}
{\pmac{}} ensures that once the dynamic schedules are generated locally, the nodes in \dnodes{} can follow those schedules to handle the disturbance without transmission collisions with other nodes in the network. Since all the nodes in \dnodes{} receive the same disturbance information, the dynamic schedules generated locally at these nodes are all consistent. The construction of a dynamic schedule must guarantee that {\extend 1) all rhythmic packets meet their timing and reliability requirements, 2) 
the reliability degradation of periodic packets is minimized,
and 3) the system can reuse the static schedule after the rhythmic mode ends and all packets can be reliably delivered by their nominal deadlines.} 


\subsection{Problem Formulation}
\label{ssec:problem}

In \fpas{}, the network starts operation by following a static schedule which guarantees that all tasks meet their {\extend timing and reliability requirements} if no disturbance occurs. The static schedule is generated at each node locally using the local schedule generation technique proposed in \cite{zhang2017distributed}. {\extend 
To satisfy the reliability requirement, the retransmission mechanism introduced in \cite{flexible, probabilistic} is employed for each task to achieve the desired PDR value, {\em i,e,}, $\lambda_{i,k} \geq \lambda^R$. In the following, we assume the network {\han adopts the} TBS model where additional time slots are assigned to individual transmissions. (The case is similar for the PBS model where slots are assigned to individual packets.)} We denote the static schedule as $\static = \{(t, i, h)\}$, where $t$ is the slot ID, $i$ is the task ID
and $h$ is the hop index. For any given time slot $t$, we have $\static[t] = (i,h)$ if $t$ is assigned to the $h$-th transmission of $\tau_i$. Otherwise, $\static[t] = (-1, -1)$ to indicate an idle slot. 
{\extend Let $\overrightarrow{R}_{i,k} = [R_{i,k}[0], R_{i,k}[1], \dots, R_{i,k}[H_i-1]]$ be the {\em retry vector} of packet $\chi_{i,k}$ used in the static schedule in which $R_{i,j}[h]$ denotes the number of slots assigned to hop $h$ of $\chi_{i,k}$. We use $w_i^+$ to denote the number of slots assigned to $\tau_i$ ({\em i.e.}, $\overrightarrow{R}_{i,k}$) in the static schedule
which guarantees the e2e PDR value $\lambda_{i,k}$ to be larger than $\lambda^R$ in the system nominal mode. }

As shown in Fig. \ref{fig:rhythmic}, when a disturbance is detected at $r_{0,m}$, $\tau_0$ requires to enter its rhythmic state from the next release time $r_{0, m+1}$, {\em i.e.}, $\tnr = r_{0, m+1}$. Then, the system enters the rhythmic mode with an increased workload induced by $\tau_0$. A dynamic schedule $\dynamic$ is thus needed before the system switches back to the nominal mode and reuses static schedule $\static$. $\dynamic$ starts from $\tnr$ and ends at a carefully chosen end point $\tep$ of the system rhythmic mode. 
To achieve guaranteed fast disturbance handling, we further define $t_{ep}^u$ as a user specified parameter which bounds the maximum allowed DHL, and is often application dependent.
Though it is natural to use idle slots in  $\static[\tnr, \tep)$ to accommodate the increased rhythmic workload, they are not always sufficient to guarantee {\extend the timing and reliability requirements} of all rhythmic packets. In this case, some periodic transmissions have to be dropped. 
Since any node $V_j \notin$ \dnodes{} keeps following the static schedule $\static$ to transmit periodic packets, periodic transmissions cannot be adjusted in the dynamic schedule\footnote{Some periodic tasks may share common nodes with $\tau_0$ on their routing paths, which indicates that the periodic transmissions at these nodes can be adjusted in the dynamic schedule. Due to page limit, we leave this discussion to our future work and focus on the case that all periodic transmissions should not be adjusted in the dynamic schedule.}.
{\extend Therefore if any periodic transmission $\pkt_{i,k}(h)$ in $\static$ is replaced by a rhythmic transmission in $\dynamic$, 
the number of elements in $\overrightarrow{R}_{i,k}$ is reduced such that the reliability of packet $\pkt_{i,k}$ is degraded. 
If the remaining number of assigned slots (denoted as $w_{i,k}$) is less than $H_i$, the timing requirement of $\pkt_{i,k}$ is also violated since at least $H_i$ slots are needed to guaratee the delivery of $\pkt_{i,k}$.
To capture the reliability degradation for periodic packet $\pkt_{i,k}$, let $\delta_{i,k}$ represent the difference between the required PDR $\lambda^R$ and the updated PDR value $\lambda_{i,k}$ in the dynamic schedule, {\em i.e.}, $\delta_{i,k} = \max\{0, \lambda^R - \lambda_{i,k}\}$. Note that the timing degradation of each packet can also be captured by $\delta_{i,k}$ where $\delta_{i,k} = \lambda^R$ if $\chi_{i,k}$ is dropped.
Then, the question is which periodic transmissions should be replaced by rhythmic transmissions to generate dynamic schedule $\dynamic[\tnr, \tep)$ such that (i) all rhythmic packets meet their timing and reliability requirements and (ii) the total reliability degradation of periodic packets is minimized.}

Formally, to satisfy Constraints (ii), (iii) and (iv) in \prob{}, we aim to solve the following two subproblems. 

\vspace{0.05in}
\noindent {{\hu \pend}{} -- End Point Selection:} Given task set \tset, $\tnr$, $t_{ep}^u$ and  static schedule $\static$, this subproblem determines the end point $\tep$  that satisfies the following two constraints.

\begin{constraint}\label{con:ep1}
$f_{0,m+R} \leq \tep \leq t_{ep}^u$
\end{constraint}

\noindent Here, $f_{0,m+R}$ is the finish time of the last packet released in $\tau_0$'s rhythmic state.
$f_{0,m+R} \leq \tep$ ensures that the current rhythmic event can be completely handled before the system switches back to the nominal mode.

\begin{constraint}\label{con:ep2}
The system can switch back to the nominal mode and reuse the static schedule from $\tep$ {\hu and all packets after $\tep$ {\extend can be reliably delivered by their nominal deadlines}}.
\end{constraint}

\noindent 
{\pdrop{} -- Dynamic Schedule Generation:} this subproblem generates the dynamic schedule $\dynamic[\tnr, \tep)$ such that {\extend the total reliability degradation of periodic packets is minimized and the following two constraints are satisfied.}

\begin{constraint}\label{cons:rhy}
All rhythmic packets meet their {\extend timing and reliability requirements.}
\end{constraint}

\begin{constraint}\label{cons:peri}
In the dynamic schedule $\dynamic[\tnr, \tep)$, any periodic transmission slot $S[t] = (i,h) (1 \leq i \leq n)$ can only either be replaced by a rhythmic transmission slot $S[t]=(0,h)$ or kept unchanged.
\end{constraint}
\noindent
Below we first discuss how \fpas{} solves the first problem.

\subsection{End Point Selection}
\label{ssec:end_point}
Determining the right end point for the dynamic schedule is vital since it impacts not only the DHL but also the number of dropped periodic packets.
A concept similar to the end point is used by OLS and is referred to as switch point~\cite{hong2015online}. Since both OLS and \fpas{} require the system to reuse the static schedule after $\tep$, to select the end point in \fpas, we borrow some ideas in OLS including aligning the actual release time of $\tau_0$ to its nominal one and reducing the number of end point candidates by only considering the actual release times of $\tau_0$. 

\fpas and OLS have two key differences for end point selection. First, to satisfy Constraint~\ref{con:ep2}, we need to determine which packets must be completed before the system reuses the static schedule at end point $\tep$. Since OLS must obey a user-specified bound on the number of adjusted transmissions in dynamic schedule \dynamic, a {\em transmission set\/} containing all transmissions to be scheduled in $\dynamic[\tnr, \tep)$ must be constructed. However, \fpas{} has no such requirement (due to its distributed nature), thus only needs to construct an \emph{active packet set} containing all packets to be scheduled. Second, according to Constraint \ref{cons:peri}, transmissions of periodic packets must not be adjusted and can only be replaced by rhythmic transmissions in the dynamic schedule. Thus, for the active packet set, we only need to consider rhythmic packets to be scheduled by $\dynamic[\tnr, \tep)$. These differences require modifications to the end point selection process, which are detailed below.

Let $\Psi(\tep)$ denote the active packet set containing all rhythmic packets to be scheduled within $[\tnr, \tep)$ {\extend and $\Phi(\tep)$ denote the periodic packet set in which each periodic packet has at least one transmission slot in the static schedule $\static[\tnr, \tep)$.}
Naturally, any rhythmic packet with both release time and deadline in $[\tnr, \tep)$ must be included in $\Psi(\tep)$. 
The question is how to treat the rhythmic packet released before $\tep$ with a deadline after $\tep$.
As shown in Fig. \ref{fig:rhythmic}, let $\pkt_{0, q^*}$ be such a packet.
To ensure the system can reuse the static schedule from $\tep$, the actual release time of $\tau_0$ must be aligned to its nominal release time after \tep. Same as OLS,
we shorten the time interval between $r_{0, q^*}$ and $r_{0, q^* + 1}$ by shifting $r_{0, q^* + 1}$ backward to the closest nominal release time of $\tau_0$, denoted as $r_{0, p^*}$. The more challenging part is adjusting the deadline and {\extend execution time of $\pkt_{0, q^*}$ since the assigned number of transmission slots} may vary depending on which hop occurs after \tep. We construct $\pkt_{0, q^*}$ by adjusting its execution time and deadline according to the position of $\tep$ by considering the following two cases.


\noindent {\bf Case 1:} If $\tep < r_{0, p^*}$, $d_{0, q^*}$ is adjusted to $\tep$. 
{\extend Suppose the first transmission slot assigned to $\tau_0$ after $\tep$ is at $t_{k_0}$ in the static schedule.
If $t_{k_0} \geq r_{0, p^*}$, it indicates that $\static[t_{k_0}]$ is the first assigned transmission slot for the first hop of $\pkt_{0, p^*}$, {\em i.e.}, $\static[t_{k_0}]=(0,1)$. Then the execution time of $\pkt_{0, q^*}$ is set to $H_0$. 
If $t_{k_0} < r_{0, p^*}$ and suppose $\static[t_{k_0}]$ is the $k_0$-th transmission slot assigned for $\chi_{0,p^*-1}$, the execution time is set to $k_0 - 1$ correspondingly.}

\noindent {\bf Case 2:} If $\tep \geq r_{0, p^*}$, {\extend suppose the first assigned transmission slot for the first hop of $\pkt_{0,p^*}$ is at $t_1$ in the static schedule,} {\em i.e.}, $\static[t_1] = (0,1) (r_{0, p^*} \leq t_1 < d_{0, p^*})$. $d_{0, q^*}$ is adjusted to $\min(\tep, t_1)$ to guarantee that the deadline of $\pkt_{0, q^*}$ is smaller than or equal to the first transmission of $\pkt_{0, q^*+1}$. Also the execution time of $\pkt_{0, q^*}$ is set to be equal to $H_0$.

Given $t_{ep}^u$, any time slot within $[f_{0,m+R}, t_{ep}^u]$ can be selected as end point $\tep$. However, to avoid checking every time instant which is time consuming, we only need to consider the actual release times of $\tau_0$ within $[f_{0,m+R}, t_{ep}^u]$ as end point candidates, denoted as $t_{ep}^c$\footnote{Such a space reduction scheme is safe and can be proved in a similar way as Lemma 2 in~\cite{hong2015online} which is thus omitted due to page limit.}. That is,
\begin{equation}\label{eq:endpoint}
\{t_{ep}^c\} = \{r_{0,k}, \forall r_{0,k} \in [f_{0,m+R}, t_{ep}^u]\}
\end{equation}

Then the dynamic schedule generation subproblem can be refined as follow.

\vspace{0.05in}
\noindent \pdrop{}: {\extend Given the end point candidate $t_{ep}^c$, active packet set $\Psi(t_{ep}^c)$, periodic packet set $\Phi(t_{ep}^c)$ and static schedule $\static[\tnr, t_{ep}^c)$, determine the dynamic schedule $\dynamic[\tnr, t_{ep}^c)$ in which the total reliability degradation of periodic packets is minimized, {\em i.e.},}
\begin{equation}\label{eq:min}
 \forall \chi_{i,k} \in \Phi(t_{ep}^c), \min \sum \delta_{i,k}.
\end{equation}
and Constraint \ref{cons:rhy} and Constraint \ref{cons:peri} are satisfied.

\section{Dynamic Schedule Generation}
\label{sec:dynamic}

{\extend In this section, we discuss how \fpas{} determines the dynamic schedule to solve \pdrop{}. For the sake of clarity, we first assume that all links in the network are reliable, i.e. $\forall e, \lambda^L_e = 100\%$. We then generalize the network model to {\han consider} lossy wireless links and extend \fpas{} to satisfy both the timing and reliability requirements of all tasks in Section~\ref{sec:reliability}.}

\subsection{Reliable Network Setting}\label{ssec:drop}

{\extend For RTWNs in which all links are reliable, $H_i$ time slots are {\han required} for each packet $\pkt_{i,k}$ to guarantee its {\han e2e delivery}.
If any of the $H_i$ transmission slots in the static schedule is replaced by a rhythmic transmission in the dynamic schedule, $\pkt_{i,k}$ cannot be delivered and {\han has to be dropped}. Thus, the objective in Eq. (\ref{eq:min}) is {\han reduced to minimize} the number of dropped periodic packets.
We use $\dropset[\tnr, t_{ep})$ to denote the dropped periodic packet set and in the following
we illustrate that determining $\dropset[\tnr, t_{ep})$ 
is a non-trivial problem by the following Lemma.}
\begin{lemma}\label{lem:np}
Given end point $\tep$, an active packet set $\Psi(\tep)$ containing all rhythmic packets of $\tau_0$ to be scheduled and a static schedule $\static[\tnr, \tep)$, determining the dropped packet set $\dropset[\tnr, t_{ep})$ with the minimum number of dropped packets and satisfying both Constraint \ref{cons:rhy} and Constraint \ref{cons:peri} is NP-hard. 
\end{lemma}
\eat{
Below, we show that this problem is NP-hard and thus the dynamic schedule generation problem \pdrop{} is also NP-hard. 

\begin{lemma}\label{lem:np}
Given end point $\tep$, an active packet set $\Psi(\tep)$ containing all rhythmic packets of $\tau_0$ to be scheduled and a static schedule $\static[\tnr, \tep)$, {\sh the packet dropping problem to determine} the dropped packet set $\dropset[\tnr, t_{ep})$ with the minimum number of dropped packets and satisfying both Constraint \ref{cons:rhy} and Constraint \ref{cons:peri} is {\hu NP-hard}. 
\end{lemma}
}

\begin{IEEEproof}
We prove the lemma by reducing the set cover problem~\cite{karp1972reducibility} to a special case of the packet dropping problem.

The set cover problem is defined as follows: Given a set of $n$ elements $X = \{x_1, x_2, \dots, x_n\}$ and a collection $C = \{C_1, C_2, \dots, C_m\}$ of $m$ nonempty subsets of $X$ where $\cup_{i=1}^m C_i = X$. The set cover problem is to identify a sub-collection $C_s \subseteq C$ whose union equals $X$ such that $|C_s|$ is minimized.

Given a set cover problem, we can construct a special case of the packet dropping problem in polynomial time as follows:

(1) Suppose that after utilizing the original transmission slots of $\tau_0$ and the idle slots in $\static[\tnr, \tep)$ to accommodate rhythmic transmissions in $\Psi(\tep)$, there still remain $n$ packets of $\tau_0$, denoted as $\{x_1, x_2, \dots, x_n\}$, to be scheduled. Each packet $x_i$ only needs one slot to transmit.

(2) In the static schedule $\static[\tnr, \tep)$, there are $m$ periodic packets, denoted as $\{C_1, C_2, \dots, C_m\}$. For each packet $C_j$, if there exists a transmission of $C_j$ falls into the time window of rhythmic packet $x_i$ ({\em i.e.}, $[r_{x_i}, d_{x_i})$), we have $x_i \in C_j$. 

Thus, one can determine the minimum number of dropped packet set $\dropset[\tnr, t_{ep})$ that can accommodate all the rhythmic packets if and only if the smallest sub-collection $C_s$ whose union equals $X$ can be identified. The Lemma is proved.
\end{IEEEproof}

After the dropped packet set is determined, the dynamic schedule can be obtained in linear time by assigning the transmissions of the rhythmic packets to the static schedule $\static[\tnr, \tep)$ using both idle slots and transmission slots of the dropped packets. Thus Lemma \ref{lem:np} readily leads to Theorem~\ref{thm:np} and the proof is omitted.
\begin{theorem}\label{thm:np}
{\extend Generating a dynamic schedule with the minimum number of dropped packets in reliable RTWNs is NP-hard.}
\end{theorem}
\eat{
\begin{proof}
According to Lemma \ref{lem:np}, {\sh the packet dropping problem, a }sub-problem of \pdrop{}, is already NP-hard. Thus \pdrop{} is also NP-hard.
\end{proof}
}

\eat{Below we focus on solving the {\sh packet dropping problem by first giving an ILP solution and then a heuristic solution.} 

key sub-problem of \pdrop{} that how to determine the dropped packet set with the minimum number of dropped packets.
}


\eat{
\subsubsection{ILP Solution}
\label{ssec:ilp}
In this section, we present the ILP formulation {\sh to solve the packet dropping problem.} 
}




Below we focus on solving the {\sh packet dropping problem}. 
{\extend 
An ILP based formulation can be derived by associating each periodic packet with a binary variable indicating whether the packet should be dropped or not. The objective is to minimize the number of dropped packets subject to the constraint that the total number of transmission slots freed from the dropped packets should be sufficient to meet the demand of all the rhythmic transmissions in $[\tnr, \tep)$. 

We introduce the following notation:
\begin{itemize}
\item $\slotset_j = [\slotvalue_j^1, \slotvalue_j^2, \dots, \slotvalue_j^n] \; (1 \leq j \leq m)$ denotes the transmission vector of periodic packet $\pkt_j$ where each $\slotvalue_j^i$ is the number of transmissions from $\pkt_j$ in the static schedule that can be replaced by transmissions of rhythmic packet $\pkt_{0,i}$ in the dynamic schedule. Specifically, transmission $\pkt_{j}(h)$ of $\pkt_j$ {\hu can be replaced by $\pkt_{0,i}$} if $\static[t] = (j,h)$ and $r_{0,i} \leq t < d_{0,i}$.

\item $l_j$ denotes the dropping decision of periodic packet $\pkt_j$. $l_j = 1$ if $\pkt_j$ is dropped. Otherwise, $l_j = 0$. 

\item $A = [a_1, a_2, \dots, a_n]$ denotes the available slot vector where each $a_i$ represents the total number of idle slots and rhythmic transmission slots in the static schedule that can be used by rhythmic packet $\pkt_{0,i}$.
\end{itemize}

To drop the minimum number of periodic packets to guarantee the timing requirements of all rhythmic packets, we have the following objective function in the ILP formulation:

\begin{equation}\label{eq:objective}
    \min 
    \sum_{\pkt_j \in \Phi} l_j
\end{equation}

Since rhythmic transmissions are at the highest priority, the deadline of each rhythmic packet $\pkt_{0,i}$ can be guaranteed only if at least $H_{0,i}$ time slots are reserved for $\pkt_{0,i}$ in the dynamic schedule.
Also, both idle slots and rhythmic transmission slots in the static schedule can be used to satisfy $\pkt_{0,i}$'s transmission demand. Therefore, objective function (\ref{eq:objective}) is subject to the following constraint.

\begin{equation}\label{eq:constraint}
    \forall 
    \pkt_{0,i} \in \Psi, \hspace*{2em}
    \sum_{\pkt_j \in \Phi} \slotvalue_j^i \cdot l_j \geq H_{0,i} - a_i
\end{equation}
}

Given that the packet dropping algorithm is to be deployed on resource-constrained device nodes and the sizes of both the rhythmic packet set ($|\Psi|$) and periodic packet set ($|\Phi|$) can become large as the network grows, we propose a greedy heuristic to {\sh solve the packet dropping problem} which is time- and space-efficient to be deployed in practical RTWNs. 
The key idea of the \greedy{} is to drop the periodic packet which contributes the maximum number of slots to all rhythmic packets. 

\eat{
\noindent $\slotset_j = [\slotvalue_j^1, \slotvalue_j^2, \dots, \slotvalue_j^n] (1 \leq j \leq m)$ denotes the transmission vector of periodic packet $\pkt_j$ where each $\slotvalue_j^i$ is the number of transmissions from $\pkt_j$ in the static schedule that can be replaced by transmissions of rhythmic packet $\pkt_{0,i}$ in the dynamic schedule. Specifically, transmission $\pkt_{j}(h)$ of $\pkt_j$ {\hu can be replaced by $\pkt_{0,i}$} if $\static[t] = (j,h)$ and $r_{0,i} \leq t < d_{0,i}$.

\noindent $A = [a_1, a_2, \dots, a_n]$ denotes the available slot vector where each $a_i$ represents the total number of idle slots and rhythmic transmission slots in the static schedule that can be used by rhythmic packet $\pkt_{0,i}$.
}

\begin{algorithm}[tb]
\caption{Greedy Heuristic for Dropping Packets}
 \label{alg:greedy}
{\textbf{Input:} $\Psi(\tep)$, $\static[\tnr, \tep)$}\\
{\textbf{Output:} $\dropset[\tnr, \tep)$}
\begin{algorithmic}[1]
\small
\STATE{ $\Phi \gets$ periodic packet set $\{\pkt_j | 1 \leq j \leq m\}$;} \label{line:periodic}
\STATE{ $\slotset_j \gets$ transmission vector $[\slotvalue_j^1, \slotvalue_j^2, \dots, \slotvalue_j^n]$ for each periodic packet $\pkt_j$;} \label{line:vector}
\STATE{ Construct $\Psi(\tep)$'s demand vector $[v_1, v_2, \dots, v_n]$ considering the idle slots and rhythmic transmission slots in $\static[\tnr, \tep)$;} \label{line:demand}
\IF {each $v_i$ equals $0$} \label{line:idles}
    \RETURN {$\emptyset$;}
\ENDIF \label{line:idlee}
\WHILE {true}
    \STATE{Add the periodic packet $\pkt_{max}$ with the maximum $\sum_{i=1}^n \slotvalue_j^i$ in $\Phi$ to $\dropset[\tnr, \tep)$;} \label{line:max}
    \STATE{$\Phi \gets \Phi \setminus \{\pkt_{max}\}$;} \label{line:remove}
    \FOR {$i \in \{1,\dots,n\}$} \label{line:demands}
        \STATE{$v_i \gets \max(0, v_i - \slotvalue_{max}^i)$;}
    \ENDFOR\label{line:demande}
    \IF{each $v_i$ equals $0$}\label{line:returns}
        \RETURN{$\dropset[\tnr, \tep)$;}
    \ENDIF\label{line:returne}
    \STATE{Update $\slotset_j$ for each $\pkt_j \in \Phi$;}\label{line:update}
\ENDWHILE
\end{algorithmic}
\end{algorithm}

\eat{
\begin{algorithm}[tb]
\caption{Dynamic Schedule Generation}
 \label{alg:dynamic}
 \small
{\textbf{Input:} $t_{ep}^u$, $\static[\tnr, t_{ep}^u)$}\\
{\textbf{Output:} $\dynamic[\tnr, \tep)$}
\begin{algorithmic}[1]
\STATE {Construct the end point candidate set $\{t_{ep}^c\}$ according to (\ref{eq:endpoint});}
\FOR {($\forall t_{ep}^c \in \{t_{ep}^c\}$)}
    \STATE {Construct $\Psi(t_{ep}^c)$; // active packet set}
    \STATE Generate a dropped packet set $\dropset[\tnr,t_{ep}^c)$ based on $\Psi(t_{ep}^c)$ using the \greedy{};
        \STATE $\xset \gets \xset \bigcup \left\{\dropset[\tnr,\tep^c)\right\}$;
\ENDFOR
\STATE{Select the $\dropset[\tnr,\tep^*)$ with the minimum number of dropped packet in $\xset$;}
\STATE{Generate dynamic schedule $\dynamic[\tnr, \tep^*)$ based on $\dropset[\tnr,\tep^*)$ and $\static[\tnr, \tep^*)$;}
\RETURN {$\dynamic[\tnr, \tep^*)$;}
\end{algorithmic}
 \end{algorithm}
}

Alg. \ref{alg:greedy} describes how the \greedy{} drops periodic packets. Given the static schedule $\static[\tnr, \tep)$, a periodic packet set $\Phi = \{\pkt_j | 1 \leq j \leq m\}$ in which each $\pkt_j$ maintains a transmission vector $\slotset_j = [\slotvalue_j^1, \slotvalue_j^2, \dots, \slotvalue_j^n]$ is constructed (Lines \ref{line:periodic}$-$\ref{line:vector}). Considering the idle slots and rhythmic transmission slots in $\static[\tnr, \tep)$, the \greedy{} constructs a demand vector $[v_1, v_2, \dots, v_n]$ for all rhythmic packets in $\Psi(\tep)$ where $v_i$ captures the number of {\hu additional} slots required by $\pkt_{0,i}$ ($v_i = H_{0,i}-a_i$) (Line \ref{line:demand}). 
If all elements in the demand vector equal $0$, which means that the idle slots and rhythmic transmission slots in the static schedule are sufficient to accommodate all rhythmic packets in $\Psi(\tep)$, no packet needs to be dropped and an empty set is returned (Lines \ref{line:idles}$-$\ref{line:idlee}). Otherwise, the heuristic drops packets in a greedy fashion as follows. In each iteration, periodic packet $\pkt_{max}$ with the maximum $\sum_{i=1}^n \slotvalue_j^i$ in $\Phi$ is added into the dropped packet set and removed from $\Phi$ (Line \ref{line:max}$-$\ref{line:remove}). Then the algorithm updates $\Psi(\tep)$'s demand vector by subtracting $\slotvalue_{max}^i$ for each $v_i$ (Lines \ref{line:demands}$-$\ref{line:demande}). If all rhythmic packets are schedulable, {\em i.e.}, each $v_i$ equals $0$, after dropping $\pkt_{max}$, the dropped packet set $\dropset[\tnr, \tep)$ is returned (Lines \ref{line:returns}$-$\ref{line:returne}). Otherwise, the transmission vector of each periodic packet $\pkt_j$ is updated according to the status of rhythmic packets (Line \ref{line:update}). Specifically, if rhythmic packet $\pkt_{0,i}$ is already schedulable, {\em i.e.}, $v_i=0$, $\slotvalue_j^i$ is set to $0$. If $0 < v_i < \slotvalue_j^i$ which means dropping $\pkt_j$ is redundant to schedule $\pkt_{0,i}$, we have $\slotvalue_j^i = v_i$. This process repeats until all rhythmic packets are schedulable and a dropped packet set is returned.

{\extend The time complexity of the packet dropping heuristic, Algorithm~\ref{alg:greedy}, is $O(n\cdot m)$ where $n$ and $m$ are the number of rhythmic and periodic packets in the dynamic schedule, respectively.}


\subsection{Unreliable Network Setting}\label{sec:reliability}


{\extend
{\han In the discussion above}, we have assumed that all links in the RTWN are reliable, i.e., $\forall e, \lambda^L_e = 100\%$. {\han With this assumption,} both timing and reliability requirements of each task can be directly satisfied when $H_i$ transmission slots are allocated for each packet and no retransmission slot is needed. Although this assumption simplifies the algorithm design and analysis, it is not realistic in real-life settings considering the lossy nature of wireless links.
Thus, {\han in this subsection} we consider unreliable links and extend \fpas{} to handle disturbance considering both timing and reliability requirements for each task.

For RTWNs containing unreliable links, a retransmission mechanism is required and each packet may be assigned multiple retransmission slots in the static schedule according to the link quality on the routing path.
After the system enters the rhythmic mode, Alg. \ref{alg:greedy} can still be applied if we do not differentiate transmission and retransmission slots allocated for each packet. That is, if any assigned slot of a periodic packet is determined to be occupied by a rhythmic transmission in the dynamic schedule, all its associated transmissions and retransmissions along the routing path will be dropped as well. However, this causes the system performance, in terms of both timing and reliability, to drop significantly {\tz since some of the dropped periodic transmissions may be kept to deliver this periodic packet. }
Then, the challenge is to determine the dropped periodic transmission set, denoted as $\dropset^*[\tnr, t_{ep})$, which leads to the minimum reliability degradation on periodic tasks ({\em i.e.}, solving the problem defined in Eq. (\ref{eq:min})).

\eat{
{\color{red} Discuss the scenarios and notations used in the static mode (e.g. retry vector and PDR functions) and explain why we don't need to consider \fpas{} in the static case.}

\noindent \textbf{Problem 1.2*}: Given the active packet set $\Gamma$, the static schedule \static{}, the PDR function \pdrfunc{} and the retry vector function \retryfunc{} of each task $\tau_i$, determine the dynamic schedule $\dynamic[\tnr, t_{ep})$ in which the total reliability degradation is minimized, i.e.,
\begin{equation}
    \forall \chi_{i,j} \in \Gamma, \min \sum \delta_{i,j}.
\end{equation}
and Constraint \ref{cons:peri} and the following constraint are satisfied.
\begin{constraint}
All rhythmic packets are reliably delivered by their deadlines.
\end{constraint}
}

Apparently, the packet dropping problem in Sec. \ref{ssec:drop}, where dropping any transmission leads the same reliability degradation $\lambda^R$, is a special case of the transmission dropping problem considering unreliable link. Thus, according to Lemma~\ref{lem:np}, the following theorem holds and the proof is {\han omitted}.

\begin{theorem}\label{thm:np_reliable}
{\extend Generating a dynamic schedule with the minimum reliability degradation, {\em i.e.} solving \pdrop{}, is NP-hard.}
\end{theorem}


Next we focus on solving the transmission dropping problem 
and propose another heuristic. 
Note that, a packet may still be delivered even if some retransmissions are replaced by rhythmic transmissions. 
Thus, instead of dropping the packet contributing the maximum number of slots in Alg. \ref{alg:greedy}, the key idea of the heuristic is to drop the periodic transmission which results in the minimum reliability degradation at each iteration. 
In the following we first describe the calculation of the reliability degradation for each transmission.

Given the PDRs of all the links along the routing path of $\tau_i$ and the retry vector $\overrightarrow{R}_{i,k}$, the reliability value of $\chi_{i,k}$, $\lambda_{i,k}$, can be derived as:
\begin{equation}
\label{equ:pdr}
    \lambda_{i,k} = \prod_{h=0}^{H_i-1}1-(1-\lambda^L_{L_i[h]})^{R_{i,k}[h]}.
\end{equation}
If a retransmission of $\chi_{i,k}$ at $h$-th hop is dropped, the updated reliability value can be readily computed using Eq. (\ref{equ:pdr}) by updating $R_{i,k}[h]$ in the retry vector. The reliability degradation, then, is the difference between the two PDR values.

Alg. \ref{alg:greedy_r} describes the generation of the dropped transmission set using the heuristic. In the initialization phase, the periodic packet set and the rhythmic demand vector are constructed (Lines \ref{line:periodic_r} - \ref{line:demand_r}), and in Lines \ref{line:idles_r} - \ref{line:idlee_r} we check whether any periodic transmission needs to be dropped in the dynamic schedule. If so, we drop periodic transmissions in a greedy manner. At each iteration, we select the periodic transmission $\pkt_{j}(min)$ with the minimum reliability degradation according to the discussion above (Lines \ref{line:tran_min}). If any time slot of $\pkt_{j}(min)$ falls into the time window of any rhythmic packet needing extra slot to transmit, it is added into the dropped transmission set and the rhythmic demand vector is updated correspondingly (Lines \ref{line:min_use}-\ref{line:vi}). Otherwise, $\pkt_{j}(min)$ is kept and cannot be selected in the future. If all rhythmic packets are schedulable, the dropped transmission set $\dropset^*[\tnr, \tep)$ is returned.

\begin{algorithm}[tb]
\caption{Transmission Dropping Heuristic}
 \label{alg:greedy_r}
{\textbf{Input:} $\Psi(\tep)$, $\static[\tnr, \tep)$}\\
{\textbf{Output:} $\dropset^*[\tnr, \tep)$}
\begin{algorithmic}[1]
\small
\STATE{ $\Phi \gets$ periodic packet set $\{\pkt_j | 1 \leq j \leq m\}$;} \label{line:periodic_r}
\STATE{ Construct $\Psi(\tep)$'s demand vector $[v_1, v_2, \dots, v_n]$ considering the idle slots and rhythmic transmission slots in $\static[\tnr, \tep)$;} \label{line:demand_r}
\IF {each $v_i$ equals $0$} \label{line:idles_r}
    \RETURN {$\emptyset$;}
\ENDIF \label{line:idlee_r}
\WHILE {true}
    \STATE{Select the periodic transmission $\pkt_{j}(min)$ with the minimum reliability degradation from all $\pkt_j$ in $\Phi$;}\label{line:tran_min}
    \IF{$\pkt_{j}(min)$ can be utilized by any rhythmic packet $\pkt_{o,i}$ with $v_i > 0$} \label{line:min_use}
        \STATE{Drop $\pkt_{j}(min)$ and update $\pkt_{j}$ in $\Phi$;}
        \STATE{$\dropset^*[\tnr, \tep) \gets \dropset^*[\tnr, \tep) \bigcup \{\pkt_{j}(min)\}$;}
        \STATE{$v_i \gets v_i - 1$;}\label{line:vi}
    \ELSE
        \STATE{$\pkt_{j}(min)$ cannot be selected;}
        \STATE{Continue;}
    \ENDIF
    \IF{each $v_i$ equals $0$}\label{line:returns_r}
        \RETURN{$\dropset^*[\tnr, \tep)$;}
    \ENDIF\label{line:returne_r}
\ENDWHILE
\end{algorithmic}
\end{algorithm}

The time complexity of the dropped transmission determination is $O(n\cdot m \cdot w^+)$ where $n$ and $m$ are the numbers of rhythmic and periodic packets in the dynamic schedule, respectively. $w^+$ is the number of slots assigned to each periodic packet in the static schedule.

Finally, with the dropped packet (transmission) set being determined, each node in \dnodes{} can readily generate the dynamic schedule to solve \pdrop{} which is summarized in Alg.~\ref{alg:dynamic}. 
According to our testbed experiments in Sec. \ref{sec:testbed}, all nodes have runtime less than 1ms (within one time slot of 10ms) to complete the dynamic schedule generation.

Note that the proposed \fpas{} framework can be readily modified to handle disturbances in networks that adopt the PBS model. The only difference appears at the selection of the periodic transmission with the minimum reliability degradation (Line \ref{line:tran_min} in Alg. \ref{alg:greedy_r}). Since time slots are allocated to each individual packet instead of transmission in the PBS model, we select the periodic {\em packet} with the minimum reliability degradation if one of the assigned slots is replaced by a rhythmic transmission. For computing the reliability value of each packet in the PBS model, readers can refer to \cite{flexible}. 
}

\begin{algorithm}[tb]
\caption{Dynamic Schedule Generation}
 \label{alg:dynamic}
 \small
{\textbf{Input:} $t_{ep}^u$, $\static[\tnr, t_{ep}^u)$}\\
{\textbf{Output:} $\dynamic[\tnr, \tep)$}
\begin{algorithmic}[1]
\STATE {Construct the end point candidate set $\{t_{ep}^c\}$ according to (\ref{eq:endpoint});}
\FOR {($\forall t_{ep}^c \in \{t_{ep}^c\}$)}
    \STATE {Construct $\Psi(t_{ep}^c)$; // active packet set}
    \IF{The network is reliable}
        \STATE {\extend Generate dropped packet set $\dropset[\tnr,t_{ep}^c)$ using Alg.~\ref{alg:greedy};}
    \ELSE
        \STATE {\extend Generate dropped transmission set $\dropset^*[\tnr,t_{ep}^c)$ using Alg.~\ref{alg:greedy_r};}
    \ENDIF
        \STATE {\extend $\xset \gets \xset \bigcup \left\{\dropset[\tnr,\tep^c) (\dropset^*[\tnr,\tep^c)\right\}$;}
\ENDFOR
\STATE {\extend Select the $\dropset[\tnr,\tep^*)$ ($\dropset^*[\tnr,\tep^*)$) with the minimum number of dropped packets (minimum reliability degradation) in $\xset$;}
\STATE {\extend Generate dynamic schedule $\dynamic[\tnr, \tep^*)$ based on $\dropset[\tnr,\tep^*)$ ($\dropset^*[\tnr,\tep^*)$) and static schedule $\static[\tnr, \tep^*)$;}
\RETURN {$\dynamic[\tnr, \tep^*)$;}
\end{algorithmic}
 \end{algorithm}
\section{Performance Evaluation} \label{sec:simulation}


In this section, we present key performance results from both testbed experiments and simulation studies to evaluate the performance of the \fpas{} framework in RTWNs. {\tz The testbed implementation is to validate the correctness of the proposed \fpas{} framework and to obtain overhead in real applications. Extensive simulations are for performance evaluation since they allow us to easily vary taskset and network specifications to study the trend. Below we first introduce the experiments from our testbed.}

\subsection{Testbed Implementation and Evaluation}\label{sec:testbed}

\begin{figure}[t]
    \centering
    \includegraphics[width=0.95\columnwidth]{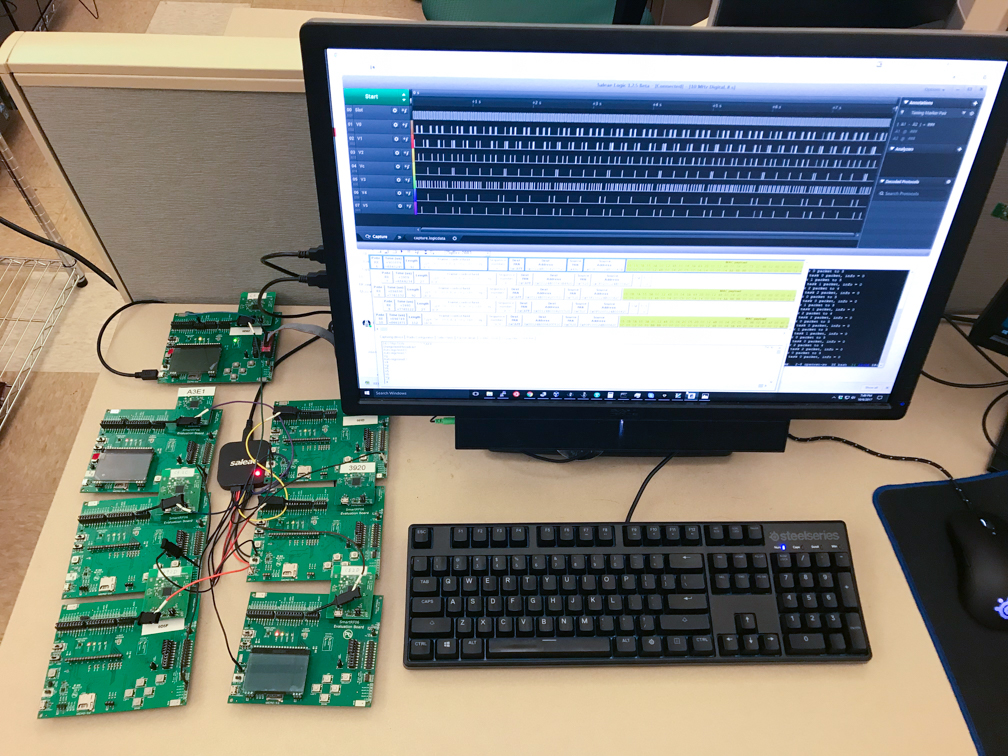}
    \caption{\small Overview of the testbed for {\fpas{}} functional validation}
    \label{fig:realtestbed}
\end{figure}

\begin{figure}[t]
    \centering
    \resizebox{0.95\columnwidth}{!}{\input{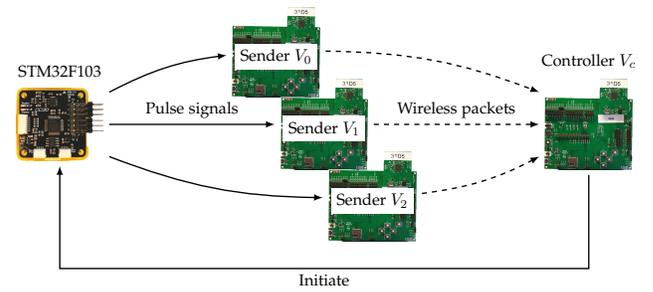}}
    \caption{Experiment setup for the measurement of latency}
    \label{fig:latencysetup}
\end{figure}
{\cam{
Our testbed is based on OpenWSN stack~\cite{watteyne2012openwsn}, an open source implementation of the 6TiSCH protocol~\cite{RFC7554}. OpenWSN enables IPv6 network over the TSCH mode of IEEE 802.15.4e MAC layer. A typical OpenWSN network consists of an OpenWSN Root and several OpenWSN devices, as well as an optional OpenLBR (Open Low-Power Border Router) to connect to IPv6 Internet. It serves as a perfect platform to experiment our proposed {\fpas{}} framework on both the data link and application layers of the stack. 

We implemented \fpas{} on our RTWN testbed to validate the correctness of the design and evaluate its effectiveness for ensuring prompt response to unexpected disturbances. The \pmac{} was implemented by enhancing the MAC layer of the OpenWSN stack and the dynamic schedule generation algorithm (using the same code as in the simulation) was implemented in the application layer. In the following, we first present the implementation of \pmac{} and its performance evaluation, and then validate the correctness of \fpas{} in a multi-task multi-hop RTWN.

As shown in Fig.\ref{fig:realtestbed}, our testbed consists of 7 wireless devices (TI CC2538 SoC + SmartRF evaluation board). One of them is configured as the root node (controller node) and the rest are device nodes to form a multi-hop RTWN. A CC2531 sniffer is used to capture the packet. A 8-Channel Logic Analyzer is used to record device activities by physical pins, in order to accurately measure the timing information among different devices. Fig.~\ref{fig:latencysetup} shows the experiment setup for the measurement of application layer performance.
}
}

\begin{table}[t]
\caption{\small Slot Timing Information of \pmac{}}
\label{tab:time_const}
\centering
\begin{tabular}{|c|c|c|c|}
\hline
Parameters & Value ($\mu$s) & Parameters & Value ($\mu$s) \\
\hline
SlotDuration & 10,000 & LongGT & 2,200 \\
TxOffset & 2,120 & ShortGT & 1,000 \\
TxAckDelay & 1,000 & PriorityTick & 30 to 400 \\
\hline 
Ext. SlotDuration & 10,800 & Extended LongGT & 3,000 \\
\hline
\end{tabular}
\end{table}

\begin{figure}[t]
    \centering
    \includegraphics[width=0.8\columnwidth]{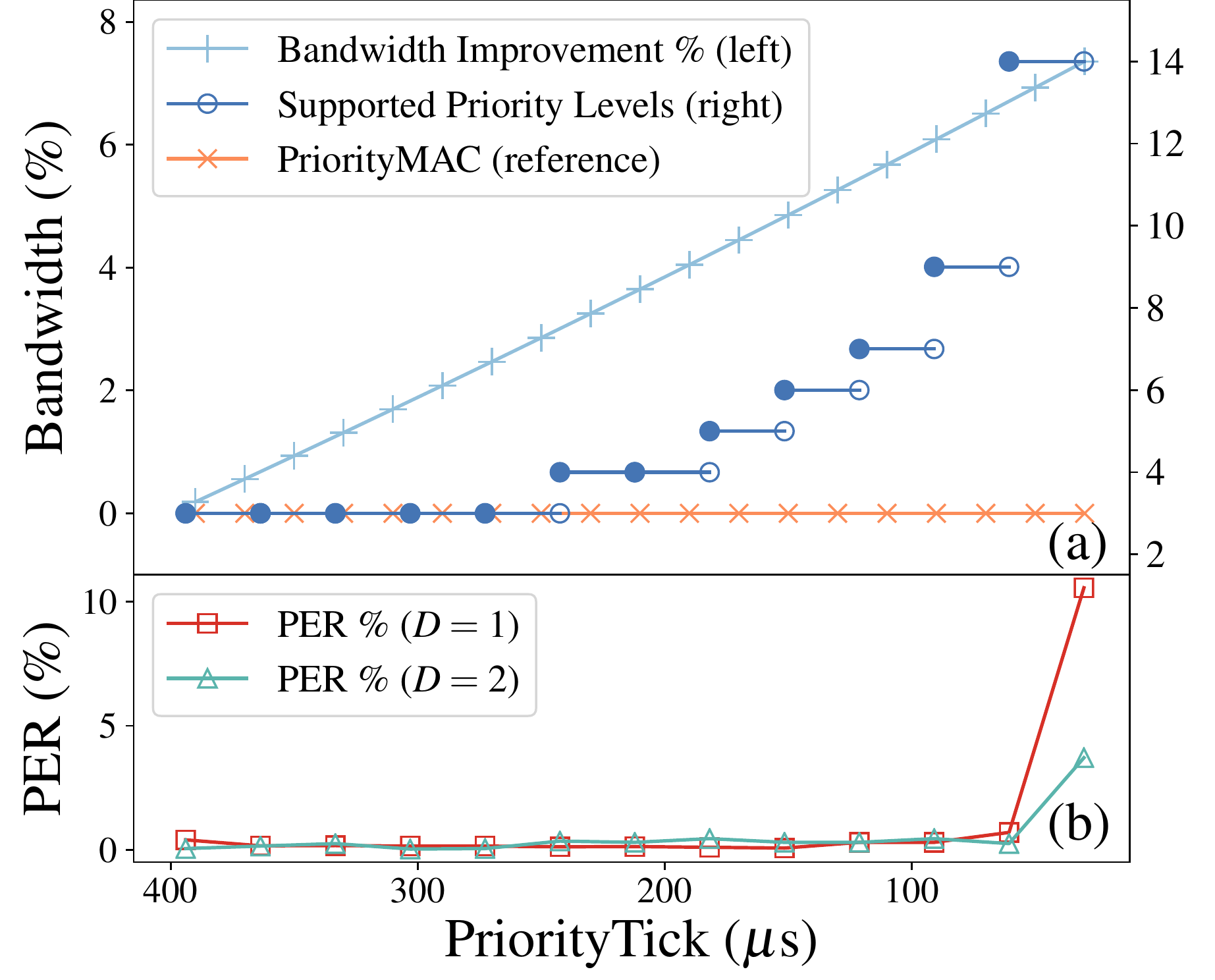}
    \caption{\small Priorities and PER vs. PriorityTick.}\label{fig:PriorityTick}
\end{figure}

\begin{figure}[t]
    \centering
    \includegraphics[width=0.8\columnwidth]{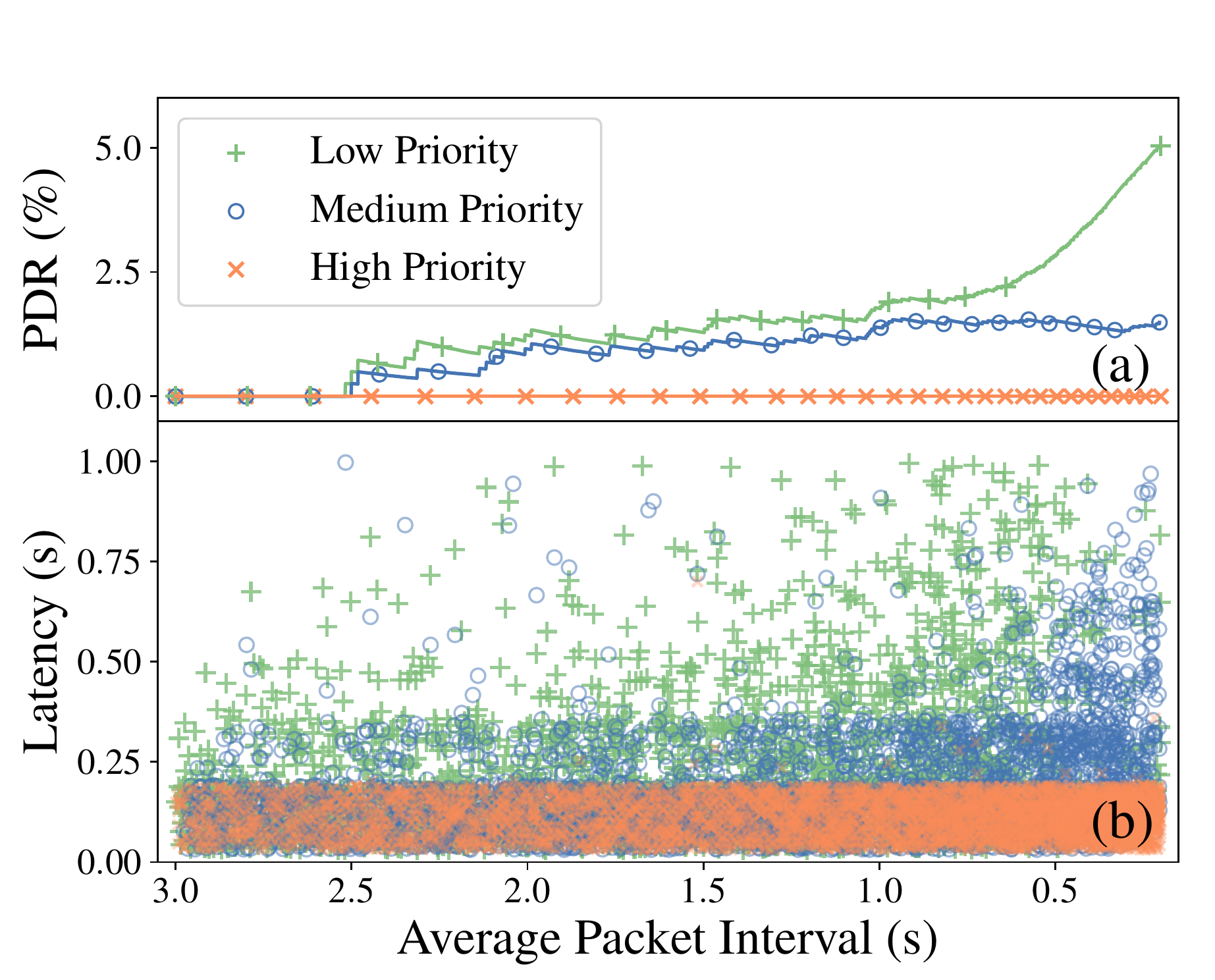}
    \caption{\small Measurement of latency and PDR.}\label{fig:latency}
\end{figure}

\subsubsection{Implementation and Evaluation of \pmac{}}
For fair comparison with PriorityMAC~\cite{prioritymac}, we used the 10ms slot timing of 802.15.4e in the \pmac{} implementation. Since PriorityMAC adds two subslots (0.4ms each) before each time slot, we also extended the SlotDuration and LongGT of {\pmac{}} by 0.8ms each. Table~\ref{tab:time_const} summarizes the slot timing of {\pmac{}}, and the \emph{Adjusted TxOffset} is computed as follows:
$$\text{\em Adjusted TxOffset} = \text{\em TxOffset} + \text{(Priority Level)} \times \text{\em PriorityTick};$$
\vspace{-0.2in}

With a given extended SlotDuration, the number of priority levels that {\pmac{}} can support, denoted as $N$, is a function of {\em PriorityTick}. In our {\pmac{}} implementation, $N$ is computed by $N = \left \lfloor{\frac{0.8ms}{\text{\em PriorityTick}}}\right \rfloor +1$. Fig.~\ref{fig:PriorityTick}(a) shows how $N$ changes when the {\em PriorityTick} varies from 30$\mu$s to 400$\mu$s with a step size of 30$\mu$s (the timer resolution in the OpenWSN stack). Compared to PriorityMAC which can only support 3 effective priority levels, {\pmac{}} can support up to 14 priority levels in theory by extending the time slot with the same amount (0.8ms). {\hu Fig.~\ref{fig:PriorityTick}(a) also illustrates the bandwidth improvement, defined as $\frac{10.8ms}{10ms + 2 \times PriorityTick} \times 100\%$, when \pmac{} only needs to maintain 3 priority levels. It can be seen that the bandwidth is improved by $7\%$ due to the reduction of the {\em PriorityTick} from $400\mu$s to $30\mu$s with 3 priority levels.}

\begin{figure*}[t]
    \centering
{\extend{}
\subfloat[Legends used in the figures\label{fig:final_validation:legend}]{
    \includegraphics[width=0.99\textwidth,page=4]{Figures/waveform.pdf}
}\\
\subfloat[Nominal mode {\han (time slot 1-60)}\label{fig:final_validation:a}]{
    \includegraphics[width=0.99\textwidth,page=1]{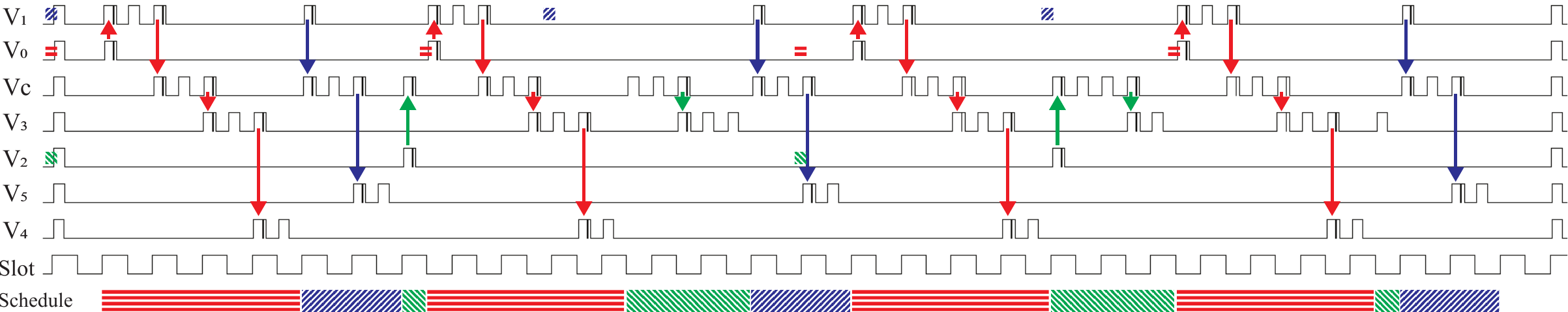}
}\\
\subfloat[Rhythmic mode using \fpas{} {\han (time slot 61-120)}\label{fig:final_validation:b}]{
    \includegraphics[width=0.99\textwidth,page=2]{Figures/waveform.pdf}
}\\
\subfloat[Rhythmic mode using r\fpas{} {\han (time slot 61-120)}\label{fig:final_validation:c}]{
    \includegraphics[width=0.99\textwidth,page=3]{Figures/waveform.pdf}
}
    \vspace{-0.05in}
    \caption{\small Slot information and radio activities in the test case captured by Logic Analyzer}
    \label{fig:final_validation}
    \vspace{-0.1in}
}
\end{figure*}

\vspace{0.025in}
\noindent {\bf Measurement of Packet Error Rate (PER):} Reducing the size of {\em PriorityTick} can support a larger number of priority levels in the RTWN system. Setting the {\em PriorityTick} too small, however, either causes nodes to lose synchronization, or make low priority senders unable to detect high priority packet transmissions and cause transmission collisions. It is thus important to identify safe {\em PriorityTick} values to make {\pmac{}} work appropriately. For this purpose, we set up a testing network with two senders talking directly to one receiver. We intentionally configure the senders to transmit in the same time slot and assign them with different priorities (using $D$ to denote the distance between the priority levels), and measure the number of correctly received packets on the receiver side. {\cam{}We define Packet Error Rate (PER) as the number of the failed transmissions divided by the number of total transmissions.} During the test, each sender generates 10,000 packets. Fig.~\ref{fig:PriorityTick}(b) shows the PER of the high priority packets by varying the size of {\em PriorityTick} from $400\mu$s to $30\mu$s. The PER of the low priority packets are always 100\%, and are thus omitted in the figure. It can be observed that {\pmac{}} works properly under most of the {\em PriorityTick} settings. Its PER only increases when the {\em PriorityTick} is reduced to $30\mu$s. This indicates that the {\pmac{}} implementation on our device node (TI CC2538 SoC) can safely support up to 9 priority levels when the {\em PriorityTick} is set to be no less than $60\mu$s. When the {\em PriorityTick} is set at $30\mu$s, it also can be observed from Fig.~\ref{fig:PriorityTick}(b) that the PER will drop (from around 10\% to 5\%) when the distance between the two priority levels increases (from $D=1$ to $D=2$).

\vspace{0.025in}
\noindent {\bf Measurement of Application Layer {\hu Performance}:} {\hu To see how \pmac{} behaves in terms of packet transmission latency and packet drop rate (PDR) for different priority levels,}
we set up a testing network with three senders and a controller node. The three senders are assigned with different priorities (high, medium and low). Their schedules are configured in a way that they transmit in the same time slot every slotframe (with a length of $165m$s). The retransmission mechanism is enabled on all the senders so that if collision happens, the failed transmission retries in the next slotframe until a maximum number of $5$ retries is reached, and the packet is then dropped. {\cam{}We define packet drop rate (PDR) as the number of dropped packets divided by the number of total packets.} We connect the controller node to a STM32F103 MCU through a UART port to control the packet generation on the senders. This STM32F103 MCU connects to the GPIO of each sender, and uses a pulse signal to trigger the sender to generate a packet. In the experiments, the controller node initiates and timestamps the packet generation. By comparing it to the timestamp of the packet reception, the application layer latency is obtained. After a successful packet reception, the controller node waits for a randomly selected time interval, and then triggers the next packet generation. {\hu To test latency and packet drop rate, w}e gradually reduce this time interval to increase the traffic volume. This will cause more transmission collisions in the network, which leads to more packet retransmissions and packet drop. 

Fig.~\ref{fig:latency}(a) and (b) show the PDR and application layer latency respectively for the three senders during the test. From the results, we observe that the packets from the high priority sender can always be transmitted in its first attempt while the medium and low priority senders have to yield upon collision by retransmission in future slotframes and suffer longer application layer latency. Similarly, when collision happens with the packets from the medium priority sender, the low priority sender has to yield again thus it is observed to have the longest latency. In Fig.~\ref{fig:latency}(a), we note that the high priority packets can always guarantee the delivery and thus its PDR is consistently 0. On the other hand, both the low priority sender and medium priority sender experience increasing packet losses when the volume of the network traffic grows, and the impact on the low priority sender is more severe. 

\subsubsection{Functional validation in a multi-task multi-hop RTWN}
We validate the correctness and effectiveness of \fpas{} by deploying it on a 7-node multi-hop network as shown in Fig.~\ref{fig:structure}. The system running in the network consists of three tasks, 
{\extend{} $\tau_0=\{\{V_0, V_1, V_c, V_3, V_4\}, 15, 8\}, \tau_1=\{\{V_2, V_c, V_3\}, 30, 6\}$ and $\tau_2=\{\{V_1, V_c, V_5\}, 20, 4\}$.} 
For each task, the first element denotes the routing path and the second one denotes its period (relative deadline). {\extend{} The third element represents the number of slots assigned to $\tau_i$, {\em i.e.} $w^+_i$, in the static schedule. We further assume that $\tau_0$ is the rhythmic task and $\overrightarrow{P_0} (\overrightarrow{D_0}) = [12, 12, 12, 12, 12]$}. The system starts running in the nominal mode at slot 1 and then switches to the rhythmic mode from slot 61. We use a Logic Analyzer to capture the radio activities from a pin of each device during time slot 1 - 120. 
In order to validate the effectiveness of \fpas{} in both reliable and lossy RTWNs, we deploy the heuristic presented in Alg.~\ref{alg:greedy} to determine the dropped packet set and use the heuristic in Alg.~\ref{alg:greedy_r} to determine the dropped transmission set, respectively. For the sake of clarity, we denote the latter as r\fpas{}.
\footnote{Both TBS and PBS models are tested. For simplicity, only the result from TBS is illustrated.}

The captured results on our testbed are illustrated in Fig.~\ref{fig:final_validation}. Specifically, {\extend{} Fig.~\ref{fig:final_validation:legend} summarizes the legends}. Fig.~\ref{fig:final_validation:a} shows the system nominal mode during time slot 1-60. Fig.~\ref{fig:final_validation:b} and Fig.~\ref{fig:final_validation:c} demonstrate the system rhythmic modes using \fpas{} and r\fpas{} during time slot 61-120, respectively.
{\extend{}
In Fig.~\ref{fig:final_validation:a}, \ref{fig:final_validation:b}, and \ref{fig:final_validation:c}, 7 waveforms represent the radio activities (transmitting, receiving, or listening) for all the 7 nodes, as labeled on the left side of the figures. Each falling or rising edge of the waveform in the \emph{Slot} row (lower part of the figures) marks the start of a new slot. In the bottom \emph{Schedule} row, slot assignments are indicated using different colors and patterns. 
Each colored small block indicates the release time of the corresponding task at a certain node. Each transmission is denoted by a colored arrow of which the starting and ending points represent the sending and receiving nodes, respectively. In the rhythmic mode, a colored circle denotes a dropped periodic transmission preempted by a rhythmic one. For example, in Fig.~\ref{fig:final_validation:a}, $\tau_1$ releases its first packet at slot $1$ and is transmitted from $V_2$ to $V_c$ at slot 15.

Fig.~\ref{fig:final_validation:a} illustrates radio activities of the system in the nominal mode (1-60 slots), after which the system switches to the rhythmic mode. 
Given by the static schedule, each packet $\pkt_{i,k}$ is allocated with extra slots for retransmission in the system nominal mode. But according to our testbed result shown in Fig.~\ref{fig:final_validation:a}, each transmission successes in its first assigned time slot without using any retransmission slot.
During the rhythmic mode (slot 61-120), task $\tau_0$ releases 5 packets as indicated in Fig.~\ref{fig:final_validation:b} and Fig.~\ref{fig:final_validation:c}. To accommodate the increased workload in the system rhythmic mode, \fpas{} determines to drop both two packets of $\tau_1$ released in the system rhythmic mode. Since the sender of $\tau_1$, {\em i.e.} $V_2$, does not receive the disturbance information, it still follows the static schedule to transmit $\tau_1$ at the assigned slots ({\em e.g.} 75). However, to ensure the transmission of the rhythmic packets, all these periodic transmissions are preempted under our designed \pmac{} mechanism (indicated by circles). 
By contrast, r\fpas{} chooses not to completely drop two packets of $\tau_1$ but to reduce the number of slots assigned to $\tau_1$'s packets both from 6 to 4. In this case, both packets of $\tau_1$ still have chances to be successfully transmitted to the destination as illustrated in Fig.~\ref{fig:final_validation:c}. This significantly increases the reliability of $\tau_1$ compared to that under the dropping decision made by \fpas{}.
These results above match those from the simulation of \fpas{} and r\fpas{} under the same experiment settings.}

\subsection{Simulation Studies}

\subsubsection{Simulation Setup}\label{ssec:setup}
In the simulation studies, we compare {\fpas{}} with both OLS and \dpas{} approaches that are able to handle unexpected external disturbances in RTWNs. 
The following two key performance metrics are used in the studies.

\noindent {\bf Success Ratio (SR)}: SR is defined as the fraction of feasible task sets over all the generated task sets. A task set is feasible only if 
a specified DRT can be achieved.

{\extend
\noindent {\bf Degradation Rate (DR)}: DR is defined as the ratio between the sum of reliability degradation from all periodic packets ({\em i.e.} $\sum \delta_{i,k}$) and the total number of generated periodic packets in the system rhythmic mode.
}

For fair comparison, we use randomly generated task sets. Each random task set is generated according to a target nominal utilization $U^*$ and by incrementally adding random periodic tasks to an initially empty set $\tset$.
The generation of each random task $\tau_i$ is controlled by the following parameters: (i) the number of hops $H_i$ drawn from the uniform distribution over $\{2, 3, \dots, 16\}$, (ii) nominal period $P_i$ drawn from the uniform distribution over $\{H_i, \dots, 500\}$, and (iii) nominal relative deadline $D_i$  equal to period $P_i$\footnote{ The unit of $P_i$ and $D_i$ values is one time slot and the range of the parameters are determined according to realistic RTWN applications.}. 
{\extend To guarantee reliable transmission, we use the TBS model to determine the number of slots assigned for each packet according to Eq. (\ref{equ:pdr}).}

After all tasks in $\tset$ are generated, we randomly select one of them as the rhythmic task $\tau_0$ and assume that the disturbance is detected at the $k$-th instance of $\tau_0$ where $k$ is randomly selected from $\{1, \dots, 20\}$. The period vector $\overrightarrow{P_0}$ ($\overrightarrow{D_0}=\overrightarrow{P_0}$) is generated by controlling the following parameters: (i) the number of elements in $\overrightarrow{P_0}$, $R$, and (ii) the initial rhythmic period ratio, $\gamma = P_{0,1} / P_0$. To better control the workload of the rhythmic task, we fix $\gamma$ to $0.2$ and tune $R$ which can be any integer in the set of $\{4,6,\dots,16\}$. Given $\gamma$ and $R$, the value of each rhythmic period $P_{0,k}$ can be computed by $P_{0,k} \, (1 \leq k \leq R) = \lfloor P_0 \times (\gamma + (k-1) \times \frac{1-\gamma}{R})\rfloor$.

Additional parameters needed are summarized as follows: 1) the maximum allowed DRT $\alpha$ which is some integer multiple of the nominal period of the rhythmic task $P_0$; 2) the end point scaling factor $\beta$ which determines the upper bound of the end point $t_{ep}^u$ where $t_{ep}^u = \trn + (\beta-1) \times P_0$. Naturally, a larger $\beta$ will lead to better performance in terms of {\extend a lower reliability degradation} but may cause longer DHL. To keep $\beta$ as small as possible without performance degradation, we set $\beta = 4$ which means the disturbance must be completely handled within 3 nominal periods after the rhythmic task returns to its nominal state. 
Other parameters used in OLS and \dpas{}, {\em e.g.} the payload size of a broadcast packet, are set to the same as that in \cite{zhang2017distributed} for fair comparison.

\eat{
\begin{figure*}[t]
\begin{minipage}[t]{0.33\textwidth}
    \centering
    \resizebox{0.9\columnwidth}{!}{\begin{tikzpicture}
\begin{axis}[
view={120}{40},
xlabel=$R$,
ylabel=$U^*$, 
zlabel=Average PDR DR (\%),
xtick distance=2,
ytick distance=0.1,
enlargelimits=false,
3d box=complete,
grid,
grid style={dashed,gray!40},
axis line style={gray!40},
legend pos=north east,
legend style={font=\small},
z buffer=sort
]
\addplot3[surf,opacity=0.35,blue,faceted color=black] coordinates {(0.5,0.5,0)};
\addplot3[surf,opacity=0.35,green,faceted color=black] coordinates {(0.5,0.5,0)};
\addplot3[surf,opacity=0.35,red,faceted color=black] coordinates {(0.5,0.5,0)};
\addplot3[surf,opacity=0.35,yellow,faceted color=black] coordinates {(0.5,0.5,0)};
\legend{OLS, FD-PaS, rFD-PaS, D2-PaS}
\addplot3[surf,opacity=0.4,blue,faceted color=black] table {Rrate_ols.dat};
\addplot3[surf,opacity=0.4,green,faceted color=black] table {Rrate_fpas.dat};
\addplot3[surf,opacity=0.4,red,faceted color=black] table {Rrate_efpas.dat};
\addplot3[surf,opacity=0.4,yellow,faceted color=black] table {Rrate_dpas.dat};
\end{axis}
\end{tikzpicture}}
    \caption{\small Comparison of the average PDR degradation rate}
    \label{fig:dr}
\end{minipage}
\begin{minipage}[t]{0.31\textwidth}
    \includegraphics[width=\textwidth]{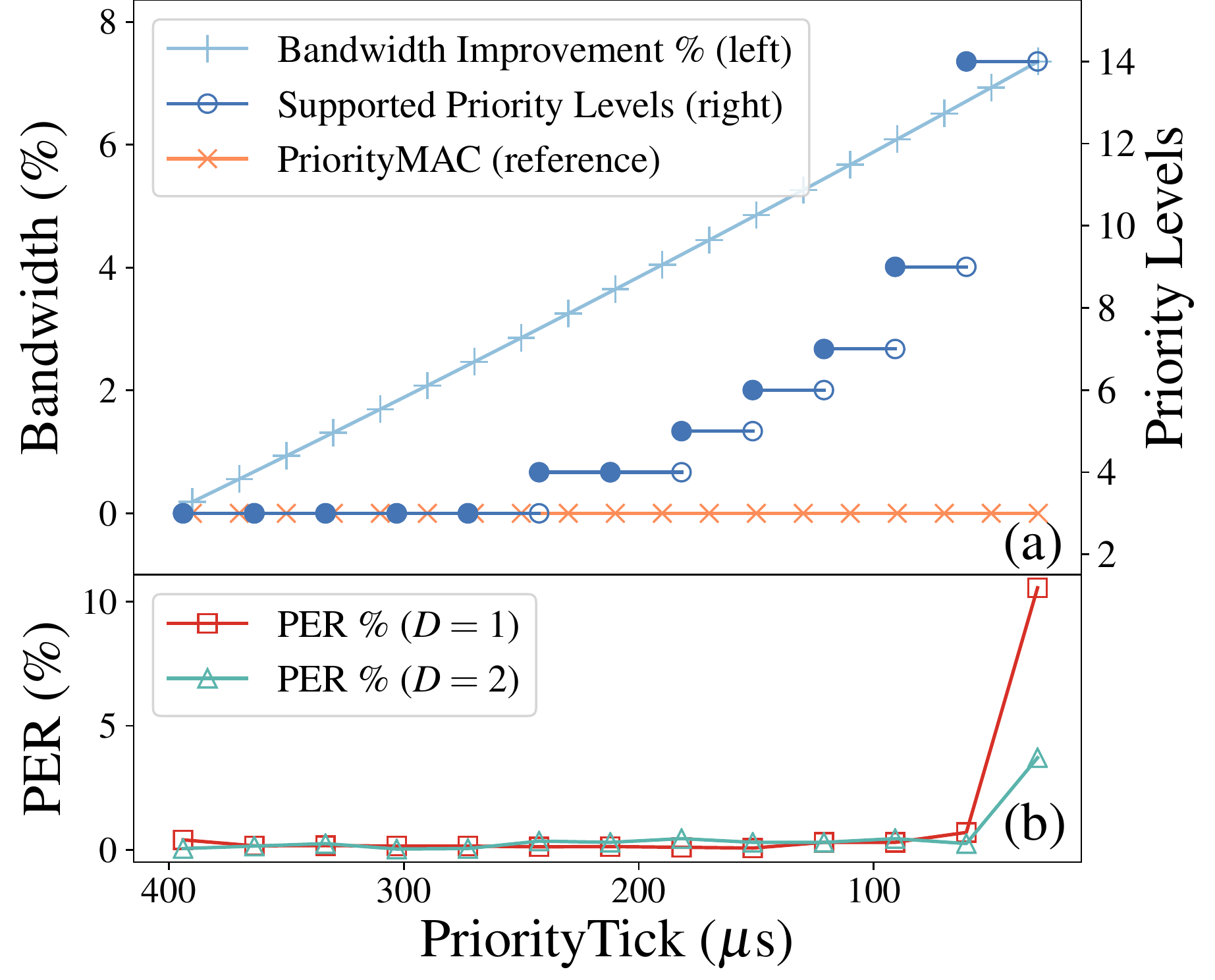}
    \caption{\small Priorities and PER vs. PriorityTick}
    \label{fig:PriorityTick}
\end{minipage}
\begin{minipage}[t]{0.35\textwidth}
    \includegraphics[width=\textwidth]{Figures/latency.pdf}
    \caption{\small Measurement of latency and PDR}
    \label{fig:latency}
\end{minipage}
\end{figure*}
}
\begin{figure}
    \centering
    \includegraphics[width=0.9\columnwidth]{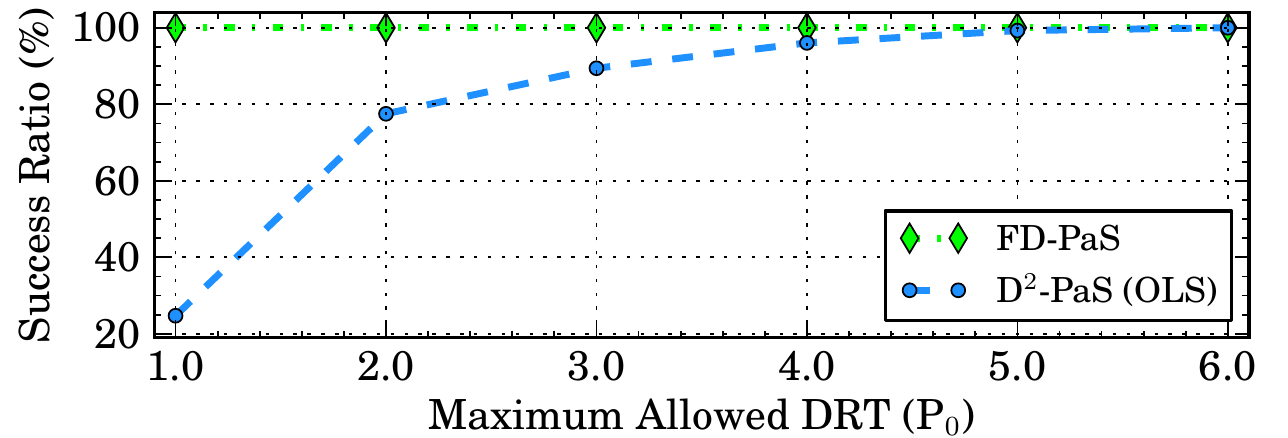}
    \caption{\small Comparison of SR with a nominal utilization $U^*=0.5$.
    }\label{fig:sr}
\end{figure}

\subsubsection{Simulation Results}
In the first set of experiments, we compare the SR of OLS, \dpas{} and \fpas{} for randomly generated task sets with a nominal utilization $U^*=0.5$ (see Fig.~\ref{fig:sr}) by varying the maximum allowed DRT, $\alpha$, from $P_0$ to 6$P_0$, with a step size of $P_0$. The results with other target nominal utilization show similar behavior and thus are omitted. In the experiments, each data point is based on 10,000 randomly generated task sets. As can be observed from Fig.~\ref{fig:sr}, \dpas{} and OLS have exactly the same SR because they both rely on a broadcast packet to propagate the disturbance information to the entire network. Under \dpas{} and OLS, the task sets are all feasible only when $\alpha=6P_0$, {\em i.e.}, the maximum allowed DRT is set to be $6$ nominal periods of the rhythmic task. However, in most practical settings, the RTWN is required to provide fast response to the disturbance within one nominal period, {\em i.e.}, $\alpha=P_0$. In this case, the SR of both \dpas{} and OLS drops to $25\%$. On the other hand, \fpas{} can always achieve $100\%$ SR since the RTWN can start handling disturbance from the beginning of the next nominal period as required by Constraint (i) of \prob{} in \fpas{}. 

In the second set of experiments, we compare the average DR of OLS, \dpas{} and \fpas{} by varying the nominal utilization $U^*$ of the randomly generated task sets and the number of rhythmic periods $R$ of the selected rhythmic task. 
{\extend As both OLS and \dpas{} do not consider unreliable links in packet scheduling, we first extend them to support reliable transmission. Specifically, all packets in OLS and \dpas{} are reliably transmitted using $w_i^+$ slots in the static schedule. In the dynamic schedule, transmission and retransmission slots assigned to each packet are not differentiated, {\em i.e.}, each packet can either be reliably scheduled or dropped.
For our proposed framework, both \fpas{} and r\fpas{} are simulated to determine the dropped packet and transmission sets in reliable and lossy RTWNs, respectively.}


\begin{figure}
    \centering
    \includegraphics[width=0.8\columnwidth]{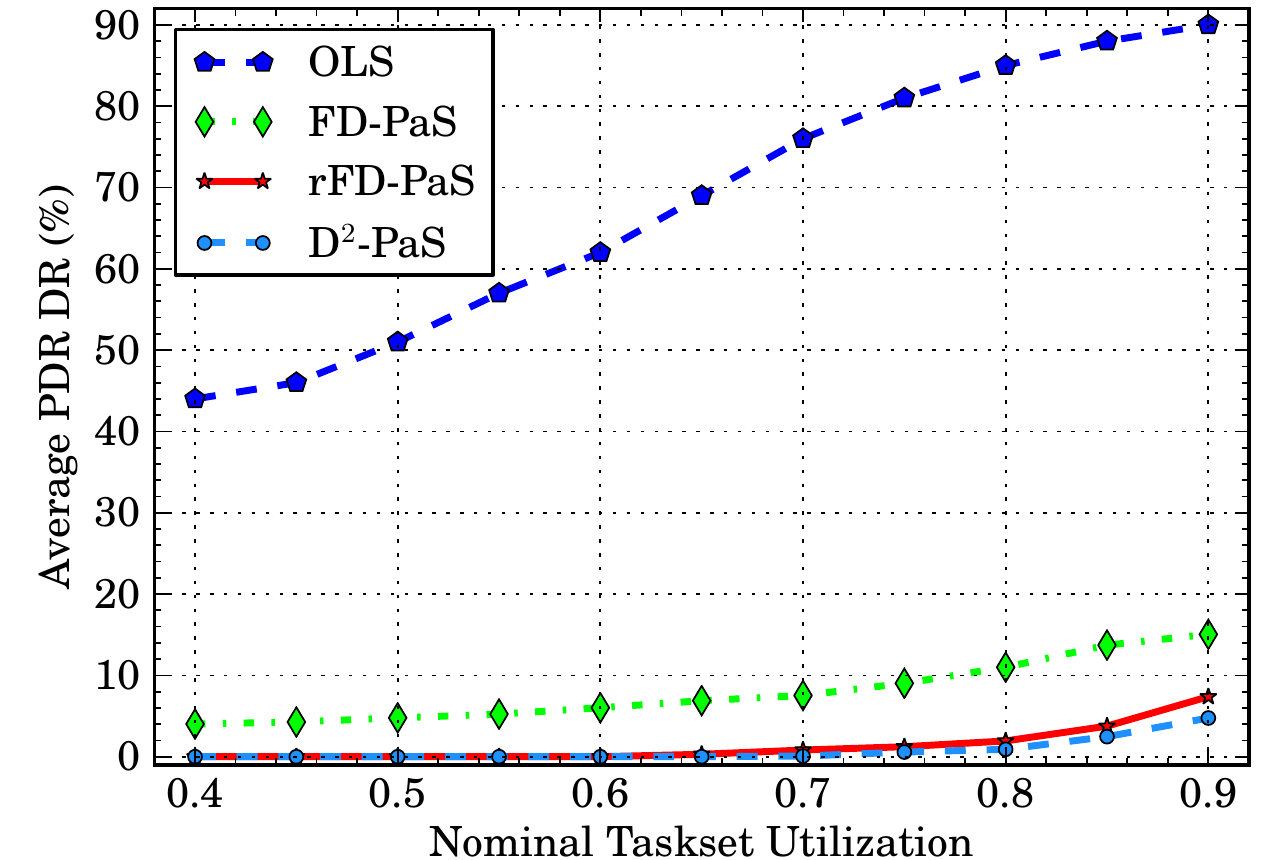}
    \caption{\small Comparison of DR with $R=10$.}\label{fig:dr_u}
    \vspace{-0.1in}
\end{figure}
\begin{figure}
    \centering
    \includegraphics[width=0.8\columnwidth]{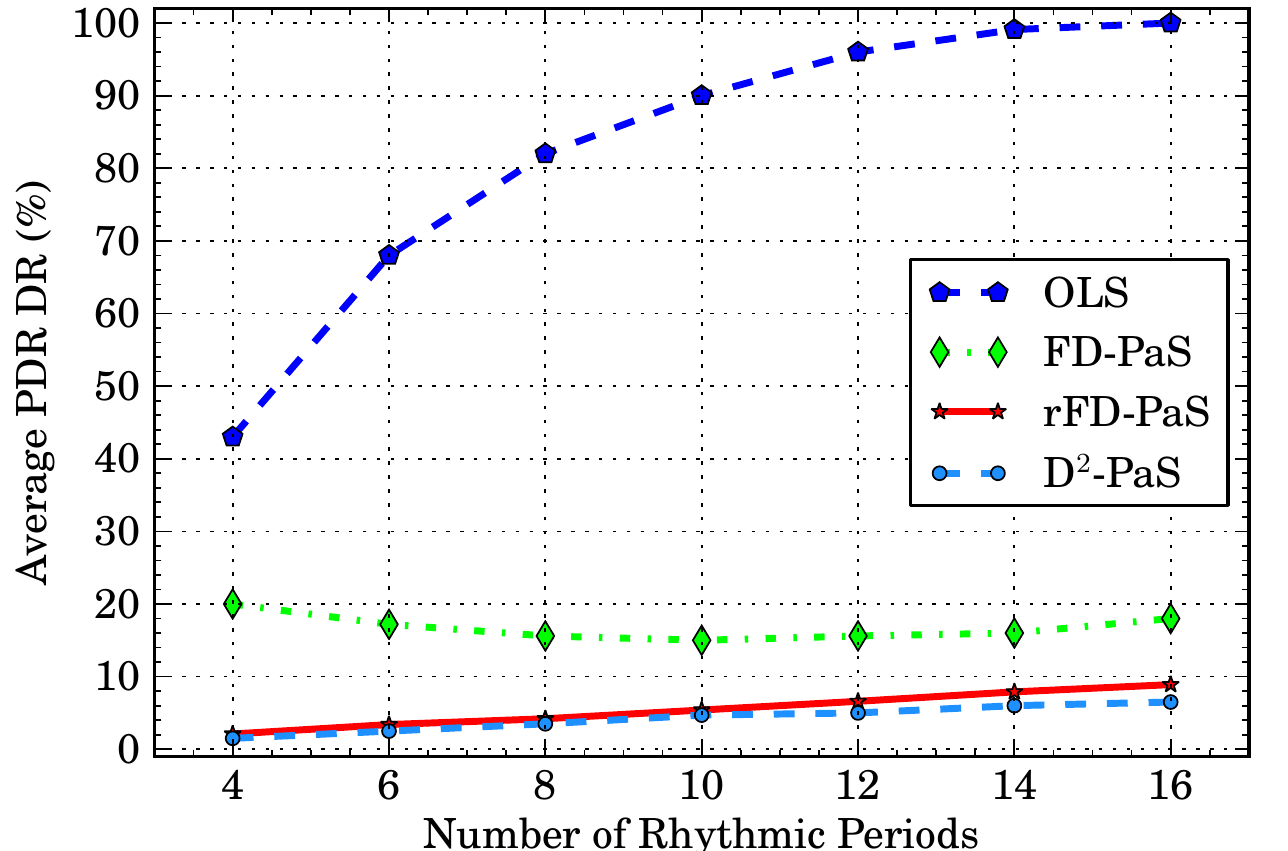}
    \caption{\small Comparison of DR with a nominal utilization $U^*=0.9$.}\label{fig:dr_r}
\end{figure}

{\extend Fig.~\ref{fig:dr_u} and Fig.~\ref{fig:dr_r} summarize the average DR as a function of $U^*$ and $R$, respectively, where each data point is the average value of $1,000$ trials.
From the figures, we can observe that both \fpas{} and r\fpas{} have significantly lower average DR over OLS ($53\%$ on average and $82\%$ in the best case). The much higher DR of OLS is due to OLS's large broadcast overhead resulted from the centralized approach.  Compared to \dpas{}, \fpas{} drops around $12\%$ more periodic packets on average since \fpas{} makes local packet dropping decisions thus tends to drop more packets. On the other hand, r\fpas{} has slightly higher average DR over \dpas{} ($1.4\%$ on average), which is contrary to our expectation since r\fpas{} has higher flexibility on adjusting the dynamic schedule. The main reason here is that the \fpas{} framework suffers from the possibility of infeasible end point selection since the system must reuse the static schedule {\tz after a well-selected end point}. Nonetheless, with the significant improvement in the average success ratio ($75\%$ when the maximum allowed DRT is set to be $P_0$), the degradation in DR is acceptable. 
}

{\han Another} observation from Fig.~\ref{fig:dr_r} is that 
the average DR of \fpas{} first drops when $R$ increases from $4$ to $10$, and then increases when $R$ keeps increasing from $10$ to $16$. 
One would expect that the average DR should monotonically increase when both $U^*$ and $R$ increase (as the case for both OLS and \dpas). Through extensive simulation studies (detailed results are omitted due to page limit), we observe that the average DR of \fpas{} is highly dependent on {\em slot utilization} 
(fraction of number of slots contributed by the dropped periodic packets, used by rhythmic packets).
When $R$ increases from a small value ({\em e.g.}, 4), the slot utilization also increases\footnote{For example, when $R$ is small, even if the rhythmic packet only needs one slot from a periodic packet, the whole periodic packet has to be dropped and all its other assigned slots are wasted.}. That is, though we need to drop more packets when the workload of the rhythmic task increases, the number of dropped packets grows more slowly than the number of packets in the system rhythmic mode. This explains why the DR of \fpas{} decreases when $R$ increases from $4$ to $10$. On the other hand, we found that when $R$ keeps increasing, the slot utilization starts decreasing, and the number of dropped packets grows faster than the total number of packets in the system rhythmic mode. This leads to the observation that the DR of \fpas{} starts to increase from $10$ to $16$.

\vspace{-0.05in}
\section{Conclusion and Future Work}
\label{sec:conclusions}


In this paper, we propose \fpas{}, a fully distributed packet scheduling framework, to handle unexpected disturbances in lossy RTWNs. Unlike centralized approaches where dynamic schedules are generated in the controller node and disseminated to the entire network, \fpas{} makes on-line decisions to handle disturbances locally without any centralized control. Such a fully distributed framework not only significantly improves the scalability but also provides guaranteed fast response to external disturbances. Our \fpas{} framework including both the multi-priority data link layer design and the dynamic schedule construction method is  implemented on our RTWN testbed. Extensive experiments have been conducted to validate its correctness and effectiveness. As future work, we will extend \fpas{} to support multi-channel settings and will explore how to handle concurrent disturbances in a fully distributed manner.

\eat{Through this paper, we make the following contributions. (i) We propose a partial disturbance propagation scheme to achieve guaranteed fast disturbance response time. (ii) To avoid transmission collisions occurred between the dynamic and static schedules, we propose a multi-priority wireless packet preemption mechanism in the data link layer to guarantee that higher-priority rhythmic packets can always get delivered. (iii) To determine a temporary dynamic schedule for handling the disturbance, we formulate a packet scheduling problem to drop minimum number of packets. By proving that the packet dropping problem is NP-hard, we propose both an optimal ILP solution and an efficient heuristic to be executed by individual nodes locally.}





\vspace{-0.1in}
\bibliographystyle{IEEEtran}
\bibliography{network}

\begin{IEEEbiography}[{\includegraphics[width=1in,height=1.25in,clip,keepaspectratio]{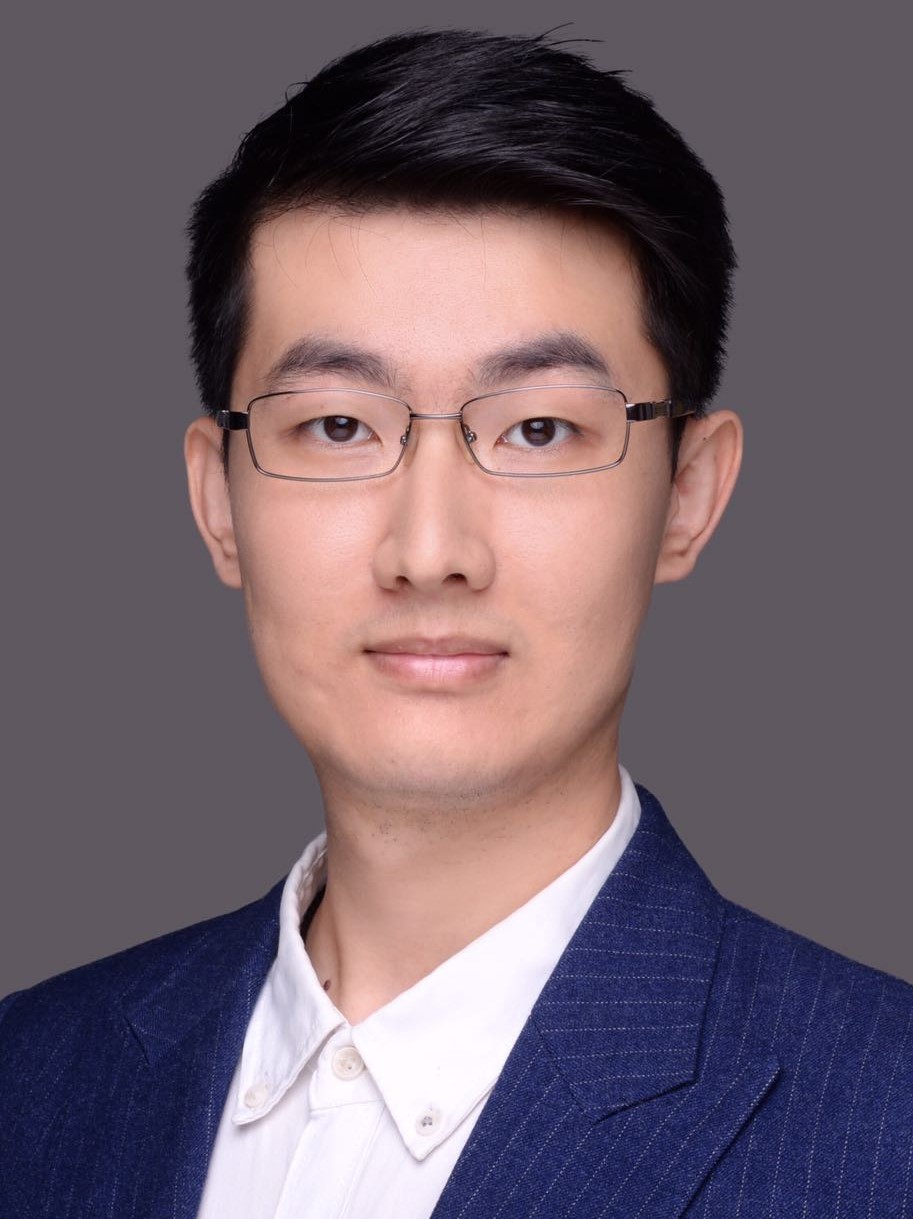}}]{Tianyu Zhang}
received the BS degree from Qingdao University in 2010, the MS degree from Northeastern Univerity, China, in 2013. He was a visiting scholar in the Department of Computer Science and Engineering at the University of Notre Dame during 2015.11 - 2017.05. His research interests include real-time systems, cyber-physical systems and wireless sensor networks. 
\end{IEEEbiography}

\vspace{-0.05in}

\begin{IEEEbiography}[{\includegraphics[width=1in,height=1.25in,clip,keepaspectratio]{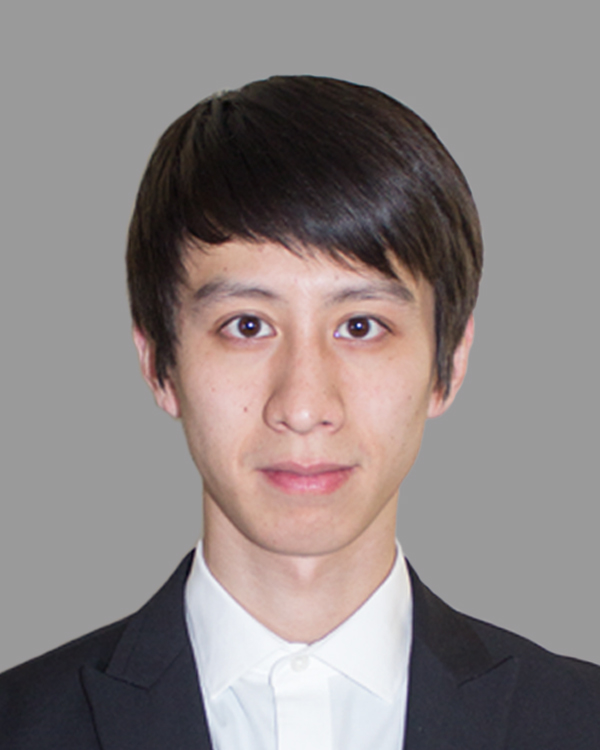}}]{Tao Gong}
received the BS degree from Beihang University, China, in 2013. He is currently pursuing the Ph.D. degree in the Department of Computer Science and Engineering at the University of Connecticut. His research interests include real-time and embedded systems, industrial wireless networks, and distributed real-time data analytics.
\end{IEEEbiography}

\vspace{-0.05in}

\begin{IEEEbiography}[{\includegraphics[width=1in,height=1.25in,clip,keepaspectratio]{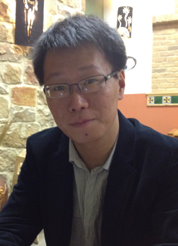}}]{Song Han}
received the BS degree from Nanjing University in 2003, the M.Phil. degree from the City University of Hong Kong in 2006, and the Ph.D. degree from the University of Texas at Austin in 2012, all in Computer Science. He is currently an assistant professor in the Department of Computer Science and Engineering at the University of Connecticut. His research interests include cyber-physical systems, real-time and embedded systems, and wireless networks.
\end{IEEEbiography}
\vspace{-0.05in}

\begin{IEEEbiography}[{\includegraphics[width=1in,height=1.25in,clip,keepaspectratio]{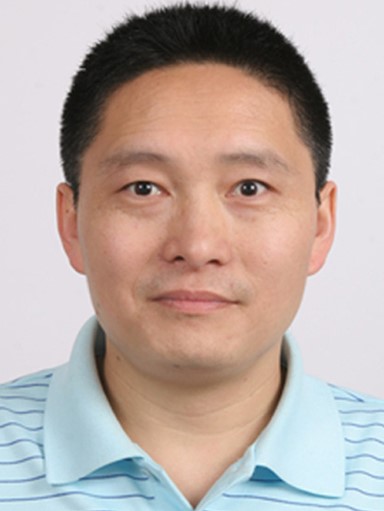}}]{Qingxu Deng}
received the Ph.D. degree in computer science from Northeastern University, Shenyang, China, in 1997. He is currently a Full Professor with the School of Computer Science and Engineering, Northeastern University, China. His research interests include reconfigurable computing systems, multiprocessor real-time scheduling, worst-case execution time (WCET) analysis and formal methods in real-time system analysis.
\end{IEEEbiography}
\vspace{-0.05in}

\eat{\begin{IEEEbiography}[{\includegraphics[width=1in,height=1.25in,clip,keepaspectratio]{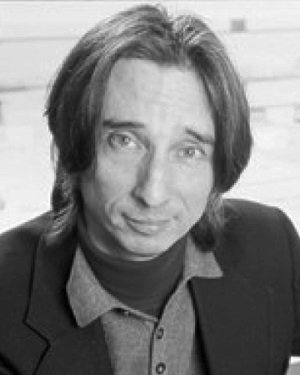}}]{Michael Lemmon}
received the B.S. degree in electrical engineering from Stanford University in 1979, and the Ph.D. degree in electrical and computer engineering from Carnegie Mellon University in 1990.
His current research interests include networked control system and the resilience of dynamical systems.
He was an Associate Editor of the IEEE Transactions on Neural Networks and the IEEE Transactions on Control Systems Technology.
\end{IEEEbiography}}

\begin{IEEEbiography}[{\includegraphics[width=1in,height=1.25in,clip,keepaspectratio]{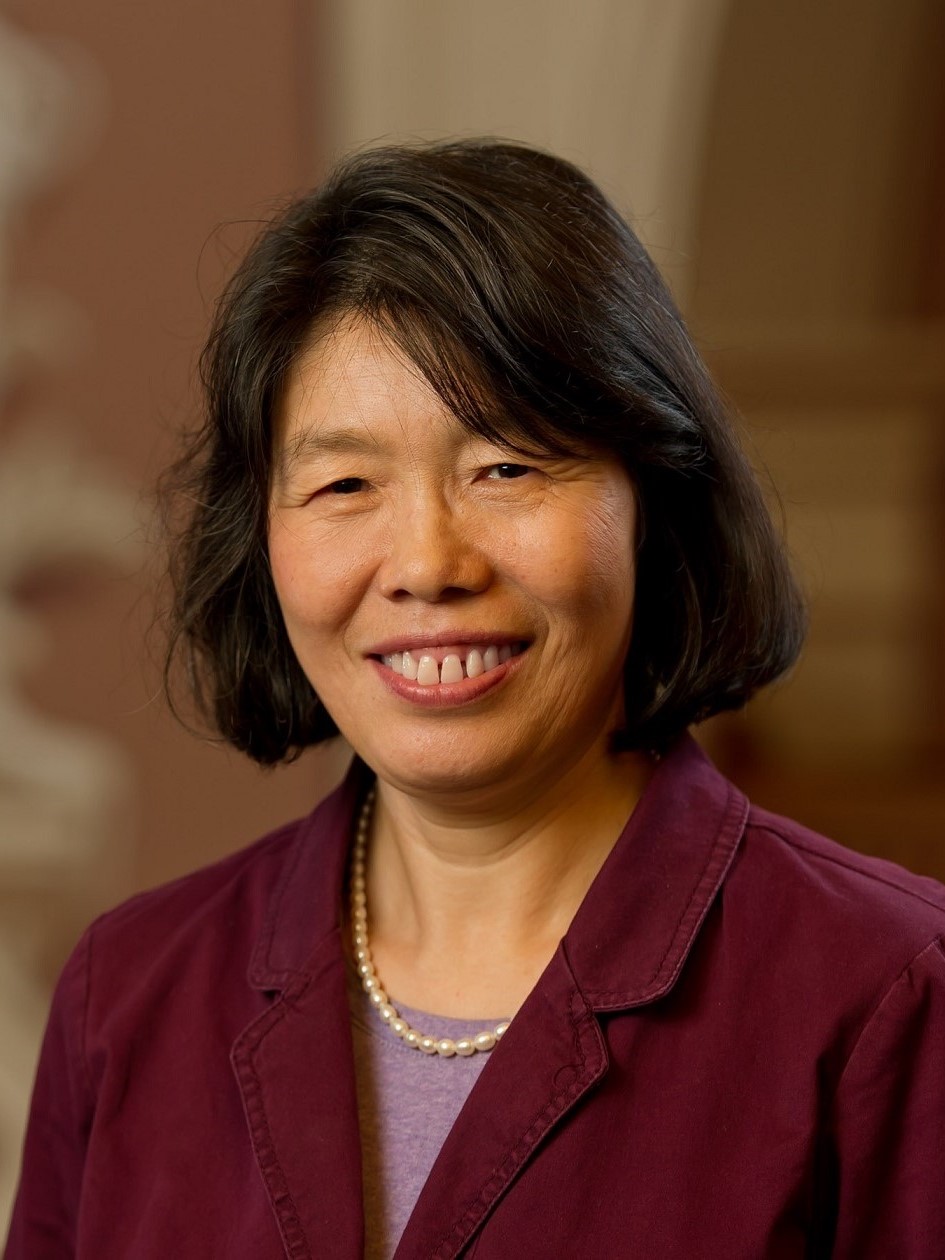}}]{Xiaobo Sharon Hu}
received the BS degree from Tianjin University, China, the MS degree from the Polytechnic Institute of New York, and the PhD from Purdue University. She is a professor in the Dept. of Computer Science and Engineering at the University of Notre Dame. Her research interests include real-time embedded systems, low-power system design, and computing with emerging technologies.
She has authored more than 250 papers in related areas. 
\end{IEEEbiography}

\end{document}